\documentclass[journal]{IEEEtran}
\usepackage{fancyhdr} 
\usepackage{amsthm}
\usepackage{amsmath}
\usepackage{amssymb}
\usepackage{cases}
\usepackage{graphicx}
\usepackage{algorithm}
\usepackage{algorithmic}
\usepackage{float}
\usepackage{multicol}
\usepackage{multirow}
\usepackage{textcomp}
\def\BibTeX{{\rm B\kern-.05em{\sc i\kern-.025em b}\kern-.08em
    T\kern-.1667em\lower.7ex\hbox{E}\kern-.125emX}}

\usepackage{cite}
\usepackage{graphicx}
\usepackage{subfigure}
\usepackage{booktabs}
\usepackage{dsfont}
\usepackage{amsfonts}
\usepackage{bm}
\usepackage{cases}
\usepackage{afterpage}
\usepackage{array}
\usepackage{boxedminipage}
\usepackage{CJK}
\usepackage{setspace}
\usepackage{threeparttable}
\usepackage{makecell}
\usepackage{diagbox}

\usepackage{chngcntr}
\theoremstyle{plain}

\usepackage[english]{babel}
\usepackage{times}
\usepackage{relsize}
\usepackage{url}
\usepackage{mparhack}
\usepackage{authblk}
\usepackage{epstopdf}
\usepackage{array}
\usepackage{hyperref}
\usepackage{caption}

\newcounter{subeqn} %
\makeatletter
\@addtoreset{subeqn}{equation}
\makeatother

\usepackage{xcolor}
\usepackage{tikz}

\newtheorem{theorem}{Theorem}
\newtheorem{lemma}[theorem]{Lemma}

\def\CN{{\mathcal N}}

\def\CI{{\mathcal I}}

\def\CB{{\mathcal B}}

\def\CK{{\mathcal K}}

\def\CT{{\mathcal T}}

\def\CG{{\mathcal G}}

\usepackage[normalem]{ulem}
\DeclareMathOperator*{\argmin}{argmin}

\graphicspath{{fig/}}

\begin{document}
\captionsetup[figure]{name={Fig.},labelsep=period}
%\title{Energy-Efficient Dynamic User Scheduling and Time Allocation for UAV Serving Systems: A Modified Actor-Critic Approach}
%\title{Actor-Critic Deep Reinforcement Learning for Energy Minimization in UAV-Aided Networks}
\title{Energy Minimization in UAV-Aided Networks: Actor-Critic Learning for Constrained Scheduling
 Optimization}

\author[ ]{Yaxiong Yuan, \textit{Student Member, IEEE},~~Lei Lei, \textit{Member, IEEE}, Thang X. Vu,  \textit{Member, IEEE}, \newline Symeon Chatzinotas, \textit{Senior Member, IEEE}, Sumei Sun, \textit{Fellow, IEEE}, and  Bj\"{o}rn Ottersten, \textit{Fellow, IEEE}
%\author[ ]{Yaxiong Yuan, Lei Lei, Thang X. Vu, Symeon Chatzinotas, Sumei Sun, and  Bj\"{o}rn Ottersten
\thanks{The work has been supported by the ERC project AGNOSTIC (742648), by the FNR CORE projects ROSETTA (C17/IS/11632107), ProCAST (C17/IS/11691338) and 5G-Sky (C19/IS/13713801), and by the FNR bilateral project LARGOS (12173206).}
\thanks{Yaxiong Yuan, Lei Lei, Thang X. Vu, Symeon Chatzinotas, and  Bj\"{o}rn Ottersten are with the Interdisciplinary
Centre for Security, Reliability and Trust, Luxembourg University, 1855 Kirchberg, Luxembourg (e-mail: yaxiong.yuan@uni.lu; lei.lei@uni.lu; thang.vu@uni.lu; symeon.chatzinotas@uni.lu; bjorn.ottersten@uni.lu).}
\thanks{S. Sun is with the Institute for Infocomm Research, Agency for Science, Technology, and Research, Singapore 138632 (e-mail: sunsm@i2r.a-star.edu.sg).}
\thanks{Part of this paper has been presented at IEEE EuCNC, June 2020 \cite{eucnc}.}
}
%\affil[ ]{\small{Interdisciplinary Centre for Security, Reliability and Trust (SnT), University of Luxembourg, Luxembourg}}
%\affil[ ]{\em{{\small Emails: \{yaxiong.yuan; lei.lei; thang.vu; symeon.chatzinotas; bjorn.ottersten@uni.lu \} }}}

%----------------------------------------------------------------------------------------------------------------------------------------------
\IEEEtitleabstractindextext{%
\begin{abstract}
In unmanned aerial vehicle (UAV) applications, the UAV's limited energy supply and storage have triggered the development of intelligent energy-conserving scheduling solutions.
In this paper, we investigate energy minimization for UAV-aided communication networks by jointly optimizing data-transmission scheduling and UAV hovering time.
The formulated problem is combinatorial and non-convex with bilinear constraints. 
To tackle the problem, firstly, we provide an optimal relax-and-approximate solution and develop a near-optimal algorithm.   
Both the proposed solutions are served as offline performance benchmarks but might not be suitable for online operation.
To this end, we develop a solution from a deep reinforcement learning (DRL) aspect.
The conventional RL/DRL, e.g., deep Q-learning, however, is limited in dealing with two main issues in constrained combinatorial optimization, i.e., exponentially increasing action space and infeasible actions.
The novelty of solution development lies in handling these two issues.
To address the former, we propose an actor-critic-based deep stochastic online scheduling (AC-DSOS) algorithm and develop a set of approaches to confine the action space.  
For the latter, we design a tailored reward function to guarantee the solution feasibility.
Numerical results show that, by consuming equal magnitude of time, AC-DSOS is able to provide feasible solutions and saves 29.94\% energy compared with a conventional deep actor-critic method.
Compared to the developed near-optimal algorithm, AC-DSOS consumes around 10\% higher energy but reduces the computational time from minute-level to millisecond-level.

\end{abstract}

}
%
%
%% make the title area
\maketitle
 \thispagestyle{fancy} % IEEE模板在\maketitle后会自动声明\thispagestyle{plain}，
                            % 导致第一页什么都没有。所以得把plain更改为fancy
      \lhead{} % 页眉左，需要东西的话就在{}内添加
      \chead{} % 页眉中
      \rhead{} % 页眉右
      \lfoot{} % 页眉左
      \cfoot{} % 页眉中
      \rfoot{\thepage} %页眉右，\thepage 表示当前页码
      \renewcommand{\headrulewidth}{0pt} %改为0pt即可去掉页眉下面的横线
      \renewcommand{\footrulewidth}{0pt} %改为0pt即可去掉页脚上面的横线

 \pagestyle{fancy}
      \rfoot{\thepage}

%
%
%% To allow for easy dual compilation without having to reenter the
%% abstract/keywords data, the \IEEEtitleabstractindextext text will
%% not be used in maketitle, but will appear (i.e., to be "transported")
%% here as \IEEEdisplaynontitleabstractindextext when the compsoc 
%% or transmag modes are not selected <OR> if conference mode is selected 
%% - because all conference papers position the abstract like regular
%% papers do.
\IEEEdisplaynontitleabstractindextext
%% \IEEEdisplaynontitleabstractindextext has no effect when using
%% compsoc or transmag under a non-conference mode.
%
%
%
%
%
%
%
%% For peer review papers, you can put extra information on the cover
%% page as needed:
%% \ifCLASSOPTIONpeerreview
%% \begin{center} \bfseries EDICS Category: 3-BBND \end{center}
%% \fi
%%
%% For peerreview papers, this IEEEtran command inserts a page break and
%% creates the second title. It will be ignored for other modes.
%\IEEEpeerreviewmaketitle
%
%
%
\section{Introduction}

Unmanned aerial vehicles (UAVs) have attracted much attention to high-speed data transmission in dynamic, distributed, or plug-and-play scenarios, e.g., disaster rescue, live concert, or sports events \cite{Mozaffariuav}.
However, UAVs' limited endurance, energy supply, and storage become critical issues for its applications, which motivates the study of energy efficiency in UAV-aided communication networks.
The UAV's energy consumption comes from two aspects, propulsion energy for flying and hovering, and communication energy for data transmission. 
The flying energy mainly depends on the UAV's velocity and trajectory \cite{Mozaffariuav}.
The hovering energy is in general proportional to the hovering time.
Compared to the propulsion energy, the communication energy consumption is not a negligible part, e.g., considerable communication energy can be consumed in the scenarios with high traffic requests from a large number of users.
Thus joint energy optimization for both parts is necessary and has attracted considerable attention in the literature \cite{Tranuav, ahmeduavenergy, ZhangJuav, zengenergyuav, zhuuavantenna, songenergyuav}.

The authors in \cite{ahmeduavenergy, ZhangJuav} maximized the energy efficiency, referring to the ratio between transmitted data and propulsion energy.
In \cite{zengenergyuav}, the authors introduced a complete UAV energy model and proposed a user-timeslot scheduling method to minimize the sum of the propulsion energy and communication energy.
Based on the energy model in \cite{zengenergyuav}, the authors formulated an energy minimization problem with latency constraints by trajectory design in \cite{Tranuav}.
The above works in \cite{ahmeduavenergy, ZhangJuav, zengenergyuav, Tranuav} adopted a time division multiple access (TDMA) mode, where the UAV serves one user per timeslot.
Besides TDMA, space division multiple access (SDMA) enables simultaneous data transmission to multiple users, such that the hovering time and hovering energy can be reduced.
In \cite{zhuuavantenna}, the authors designed an SDMA-based beamforming scheme to minimize the total transmit power for multi-antenna UAVs. 
In \cite{songenergyuav}, an energy efficiency maximization problem was investigated in an SDMA-based multi-antenna UAV network via optimizing the flying velocity and power allocation. 
However, serving multiple users simultaneously may lead to strong inter-user interference and may require more communication energy to fulfill users' demands.

Deterministic optimization algorithms, e.g., \cite{Tranuav, ahmeduavenergy, ZhangJuav, zengenergyuav, zhuuavantenna, songenergyuav} might not be suitable for fast decision making in a dynamic wireless environment.
To address this issue, deep learning-based solutions have been investigated in the literature.
The authors in \cite{jianguavdr} applied a deep neural network (DNN) for UAV-enabled hybrid networks to efficiently predict the resource allocation scheme.
In \cite{Shindluav}, a deep learning-based auction algorithm was proposed to determine a dynamic battery charging scheduling for UAV-aided systems.
Supervised learning, such as DNN, requires large amounts of training data, which is a non-trivial task in an offline manner \cite{leischeduling}.
Another category of studies is deep reinforcement learning (DRL), with the following advantages.
Firstly, DRL provides timely solutions, adapted to environment variations.
Secondly, DRL integrates DNN to make decisions and improve solution quality.
Thirdly, DNN requires an offline data generating and training phase, whereas DRL is less needed for prior knowledge and is able to train by exploring unknown environments and exploiting received feedbacks in an online manner.
In \cite{Liu2uavenergy}, the authors applied a deep Q network (DQN) to design an energy-efficient flying trajectory scheme for UAV-aided networks.
In general, DQN is used to deal with a relatively small and discrete action space, where the action space refers to the set of all possible decisions \cite{IntroRL}. 
The authors in \cite{Hedqnlarge} designed a different deep Q-learning architecture with a high dimensional action space, but it needs to evaluate all of the actions before making a decision, which is time-consuming.

Deep actor-critic is an emerging DRL method with fast convergent properties and the capability to deal with a large action space \cite{Kondaac}.  
In \cite{liuuav}, an actor-critic-based DRL (AC-DRL) algorithm was proposed to reduce the UAV's energy consumption and enhance the UAV's coverage of ground users via optimizing UAV's flying direction and distance.
In \cite{Liuuavenergy2}, the authors employed deep actor-critic to design a learning algorithm for UAV-aided systems, considering energy efficiency and users' fairness.
Note that the AC-DRL in \cite{liuuav, Liuuavenergy2} was developed for unconstrained problems.
However, most of the problems in UAV systems are constrained and with discrete variables.
%However, in many cases, AC-DRL is adopted to solve the problems with mathematical optimization models.
The conventional AC-DRL algorithms have limitations on tackling constrained combinatorial optimization problems, which may result in slow convergent, infeasible, and degraded solutions.
The authors in \cite{Caouavenergy} developed an AC-DRL algorithm for a combinatorial optimization problem in a UAV-aided system, but when the size of the action space grows exponentially, the convergence of the algorithm deteriorates.
In \cite{eucnc}, the authors applied a conventional AC-DRL approach to address an energy minimization problem in UAV networks, where the performance is limited by feasibility guarantee and rapidly-increasing action space.

In this study, we minimize the UAV's communication and propulsion energy in a downlink UAV-aided communication system.
The novelty of solution development lies in two aspects. 
Firstly, compared to offline optimization approaches, we provide online learning and timely energy-saving solutions based on DRL. 
Secondly, unlike the conventional DRL or AC-DRL methods, the proposed solution is designed to address the challenging issues in constrained combinatorial optimization. 
The major contributions are summarized as follows: 

\begin{itemize}
\item 
We formulate an energy minimization problem for an SDMA-enabled UAV communication system, where user-timeslot allocation and UAV's hovering time assignment are the coupled optimization tasks. 
The formulated problem is combinatorial and non-convex with bilinear constraints.
\item
We provide a relax-and-approximate method to approach the optimum. 
That is, the bilinear terms are addressed by McCormick envelop relaxation, then the remaining integer linear programming problem is solved by the branch-and-bound (B\&B) algorithm.
\item
We characterize the interplay among communication energy, hovering time, and hovering energy. Based on the derived analytical results, we develop a golden section search-based heuristic (GSS-HEU) algorithm for benchmarking general instances with lower complexity than the optimal solution. 
\item
Being aware of the issues in optimal/sub-optimal and conventional DRL approaches, we propose an actor-critic-based deep stochastic online scheduling (AC-DSOS) algorithm, where the original problem is transformed to a Markov decision process (MDP). 
Unlike conventional AC-DRL solutions, in AC-DSOS, we design a set of approaches, e.g., stochastic policy quantification, action space reduction, and feasibility-guaranteed reward function design, to specifically address the constrained combinatorial problem.   
\item 
Simulations demonstrate that the proposed AC-DSOS enables a feasible, fast-converging, and dynamically-adaptive solution. 
The designed approaches are effective in reducing action space and guaranteeing feasibility.
AC-DSOS achieves 29.94\% and 52.51\% energy reduction compared with a conventional AC-DRL method and a heuristic user scheduling method with almost the same computation time.
\end{itemize}

The rest of the paper is organized as follows.
Section \ref{sec:2} provides the system model and Section \ref{sec:2.2} formulates the considered optimization problem.
In Section \ref{sec:3}, we analyze the relationship between the energy consumption and hovering time, and propose a heuristic algorithm.
In Section \ref{sec:4}, we reformulate the problem as an MDP and develop an AC-DSOS algorithm.
Numerical results are presented and analyzed in Section \ref{sec:5}.
Finally, we draw the conclusions in Section \ref{sec:6}.

\textit{The codes for generating the results are online available at the link: https://github.com/ArthuretYuan}.
%--------------------------------------------------------------------------------------------------------------------------------------

\section{System Model and Problem Formulation}\label{sec:2}
\subsection{System Model}

We consider a downlink UAV-aided communication system. 
A UAV serves as an aerial base station (BS) to deliver data to ground users, e.g., for the scenarios if terrestrial BSs are unavailable or overloaded by high traffic demand from numerous users.
We assume that the UAV is equipped with $L$ antennas and each ground user has a single antenna \cite{songenergyuav}.
The UAV is fully loaded with data and energy at a dock station before the task starts.
The service area is divided into $N$ clusters considering the UAV's limited coverage area.
This setup can be used in many practical scenarios such as emergency rescue and temporary communication\cite{Zhangcluster, Yinuavcluster}.
We denote $\CN=\{1,...,n,...,N\}$ as the set of clusters and $\CN^{+}=\CN \cup \{N+1\}$ as the extended set, where the $(N+1)$-th cluster denotes the dock station.
The UAV flies through all the clusters successively according to a pre-optimized trajectory, and transmits data to the users by hovering at a given point, e.g., above the cluster's center. 
%through all the clusters successively at a fixed altitude with a predetermined flying path.
%We adopt a fly-hover-communication protocol that the UAV communicates with the ground users only when hovering above the cluster \cite{zengenergyuav}. 
Let $K_n$ and $\CK_n$ denote the number and set of the users in the $n$-th cluster.
The demands of user $k \in \CK_n$ are denoted by $q_{k,n}$ (in bits).
When all the demands in a cluster are satisfied, the UAV leaves the current cluster and visits the next one.
After serving all the clusters, the UAV flies back to the dock station.
The process of the UAV from leaving to returning the dock station is defined as a round or a task.
%When new ground users join the service area, we assume that their requests will be processed in the next round.
Fig. \ref{fig:UAV} illustrates an example of the considered system.

\begin{figure}[h]
\begin{center}
\centering
\vspace{-0.2cm}  %调整图片与上文的垂直距离
\includegraphics[scale=0.58]{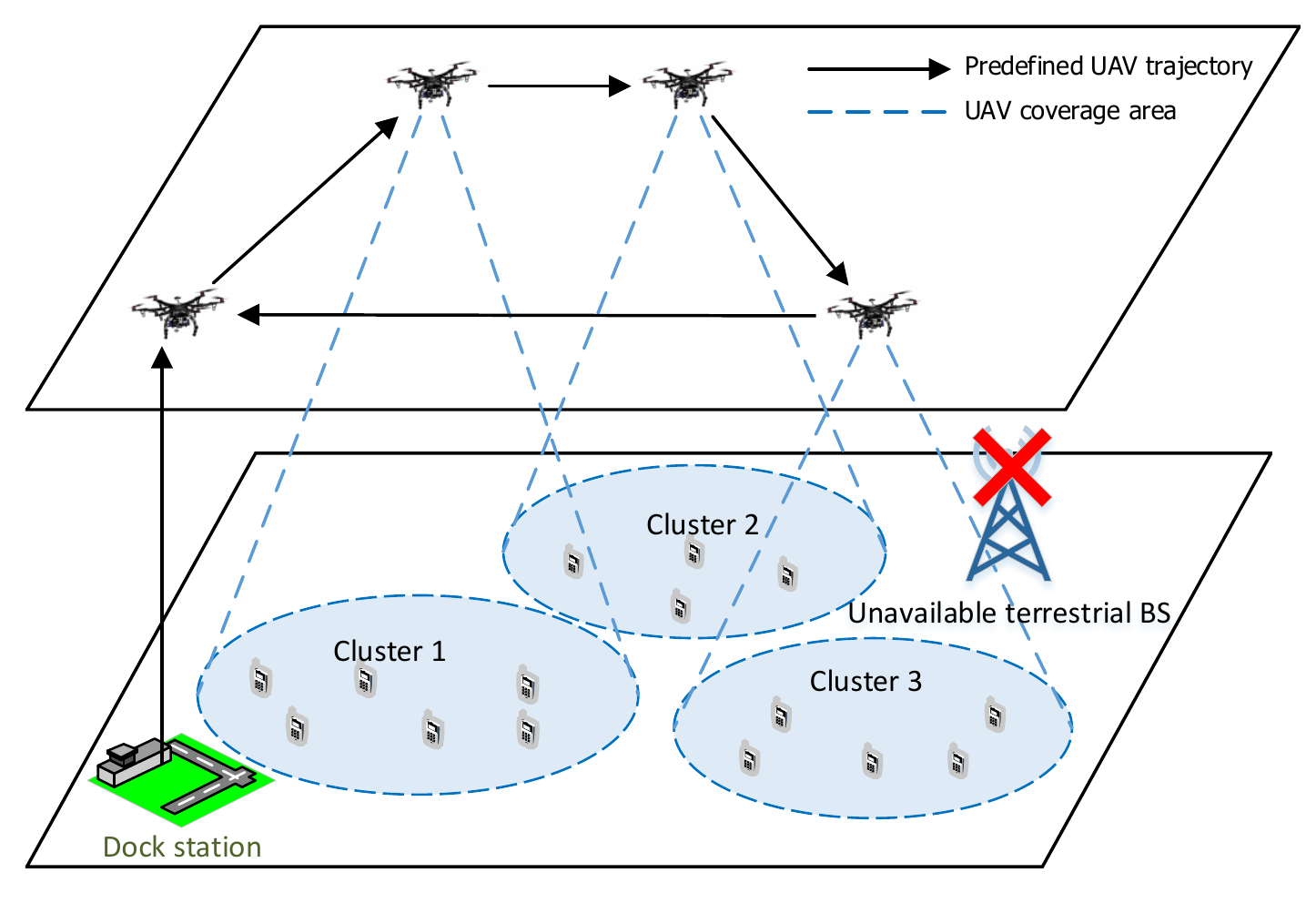}
\setlength{\abovecaptionskip}{0.2cm}   %调整图片标题与图距离
\setlength{\belowcaptionskip}{-0.6cm}   %调整图片标题与下文距离
\captionsetup{font={small}}
\caption{An illustrative UAV-aided network.}
\label{fig:UAV}
\end{center}
\end{figure}

The data stored in the UAV typically has a certain life span \cite{lifespan}.
Thus, we consider the transmitted data is delay-sensitive, and all data delivery must be completed within $T_{max}$ (in frames), where the time domain is divided by frames in set $\CT = \{1,...,t,...,T_{max}\}$.
One frame consists of $I$ timeslots, and the duration of a timeslot is $\Phi$.
With SDMA, the UAV can simultaneously transmit data to more than one user in each timeslot.
The frame-timeslot structure is shown in Fig. \ref{fig:schduling}, where the shaded blocks indicate that the users are scheduled.
We define the scheduled users at a timeslot as a user group. 
%Compared with TDMA, SDMA is more flexible but it also brings about an exponential explosion with regard to the number of groups.
The union of the possible groups in cluster $n$ is denoted by $\mathcal{G}_n=\{1,...,g,...,G_n\}$.
%, where $g=1$ refers to an empty group, i.e., the group without scheduled users.
The maximum number of candidate groups in cluster $n$ is $G_n = 2^{K_n}-1$ \cite{yuandnn}, which increases exponentially with $K_n$.
The number and set of the users of group $g$ in cluster $n$ are denoted by $K_{g,n}$ and $\mathcal{K}_{g,n}$, respectively.

\begin{figure}[h]
\begin{center}
\centering
\vspace{-0.2cm}  %调整图片与上文的垂直距离
\includegraphics[scale=0.55]{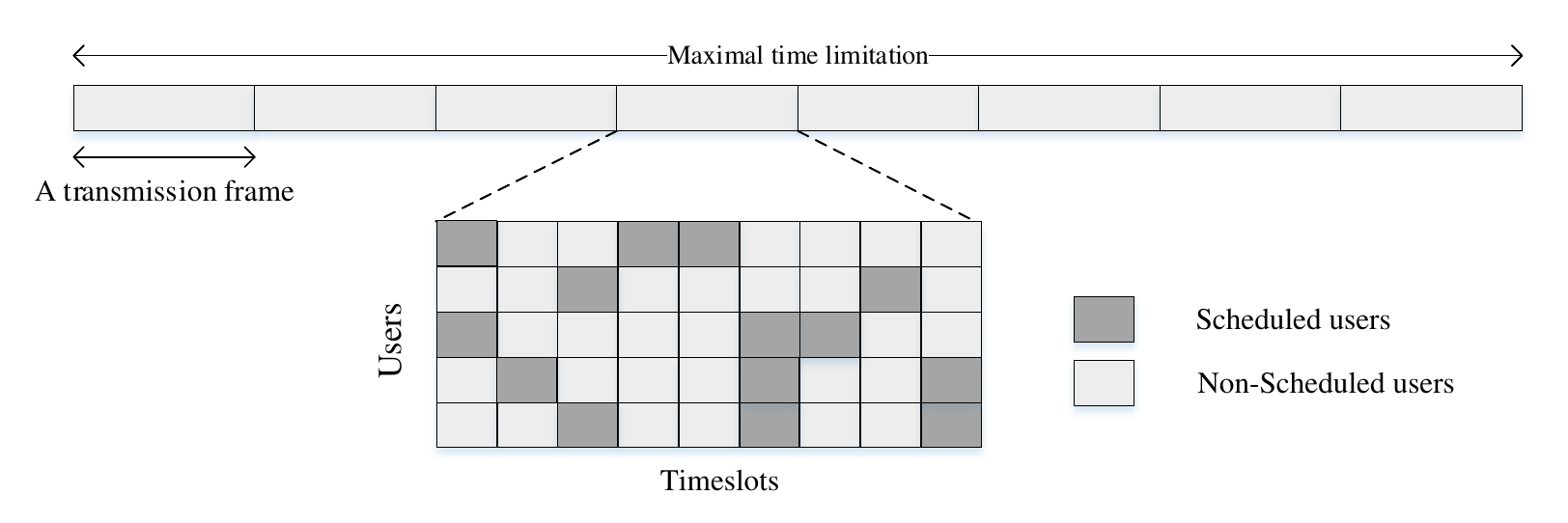}
\setlength{\abovecaptionskip}{-0.1cm}   %调整图片标题与图距离
\setlength{\belowcaptionskip}{-0.2cm}   %调整图片标题与下文距离
\captionsetup{font={small}}
\caption{An illustration of the frame-timeslot structure.}
\label{fig:schduling}
\end{center}
\end{figure}

We consider a quasi-static Rician fading channel which comprises both a deterministic line-of-sight (LoS) component and a random multipath component \cite{Youricianuav}.
The channel states are static within a transmission frame, and varying from one frame to another.
The channel vector from the UAV antennas to ground user $k \in \CK_n$ is denoted as $\mathbf{h}_{k,n}\in \mathbb{C}^{1\times L}$, which can be expressed by $\boldsymbol{\alpha}_{k,n}10^{-\xi_{k,n}/10}$, where $\boldsymbol{\alpha}_{k,n} \in \mathbb{C}^{1\times L}$ is the multipath Rician fading vector and $\xi_{k,n}$ is the free-space propagation loss between the UAV and ground user $k \in \CK_n$.
We collect all the channel vectors of the users in $\mathcal{K}_{g,n}$ to form a matrix $\mathbf{H}_{g,n} \in \mathbb{C}^{K_{g,n} \times L}$.
Within a user group, we apply a linear minimum mean square error (MMSE) precoding scheme due to its high efficiency and low computational complexity in mitigating intra-group interference.
The precoding vector for user $k \in \mathcal{K}_{g,n}$ is calculated by:
\begin{align}
\mathbf{w}_{k,g,n}=\sqrt{p_{k,g,n}}\frac{\tilde{\mathbf{h}}_{k,g,n}}{\|\tilde{\mathbf{h}}_{k,g,n}\|},
\end{align}
where $p_{k,g,n}$ is the transmit power for user $k$ in group $g$, $\tilde{\mathbf{h}}_{k,g,n}$ is the $k$-th column in $\mathbf{H}_{g,n}^{\textit{H}}(\sigma^2\mathbf{I}+\mathbf{H}_{g,n}\mathbf{H}_{g,n}^{\textit{H}})^{-1}$, and $\sigma^2$ is the noise power.
Note that transmit power $p_{k,g,n}$ is fixed as parameters in this work by following practical UAV applications, e.g., constant transmit power can be selected from 0.1 W to 10 W \cite{Yansurveyuav}.
The signal-to-interference-plus-noise ratio (SINR) for the user $k \in \mathcal{K}_{g,n}$ is given by:
\begin{align}
\label{SINR}
&\Gamma_{k,g,n}=\frac{\beta^{(kk)}_{g,n} p_{k,g,n}}{\sum_{j\in \CK_{g,n}\setminus\{k\}}\beta^{(kj)}_{g,n} p_{j,g,n}+\sigma^2}, \notag\\
& k \in \CK_{g,n},~g\in \CG_n, 
\end{align}
where $\beta^{(kk)}_{g,n}=|\mathbf{h}_{k,n}\tilde{\mathbf{h}}_{k,g,n}|^2$ and $\beta^{(kj)}_{g,n}=|\mathbf{h}_{k,n}\tilde{\mathbf{h}}_{j,g,n}|^2$ are the effective channel gains.
Since the channel states vary over frames, we use $\Gamma_{k,g,n,t}$, $\beta^{(kk)}_{g,n,t}$ and $\beta^{(kj)}_{g,n,t}$ to track SINR and channel coefficients on the $t$-th frame.
In this work, the time-varying channel is further modeled as a first state Markov channel (FSMC).
Under the FSMC, we quantify each coefficient $\beta^{(kk)}_{g,n,t}$ and $\beta^{(kj)}_{g,n,t}$ to multiple Markov states and obtain a transition probability such that the variations of $\beta^{(kk)}_{g,n,t}$ and $\beta^{(kj)}_{g,n,t}$ follow a Markov process between frames \cite{FSMC}.

If group $g \in \CG_n$ is scheduled at timeslot $i$ on frame $t$, the amount of data transmitted to user $k \in \CK_{g,n}$ and the consumed communication energy of group $g \in \CG_n$ can be expressed by:
\begin{align}
\label{kg_rate}
&d_{k,g,n,t} = \Phi B \log_2\left(1 + \Gamma_{k,g,n,t} \right),\notag\\ 
&k \in \CK_{g,n},~g\in \CG_n,~t\in \CT,
\end{align}
and
\begin{align}
\label{pgn_power}
e_{g,n,t} = \Phi \sum_{k \in \CK_{g,n}} \beta^{(kk)}_{g,n,t} p_{k,g,n},\,\,\, g\in \CG_n,~t\in \CT,
\end{align}
where $B$ is the system bandwidth. 
Note that within a  frame, we assume a user's channel condition is identical across all the timeslots, thus index $i$ is omitted in $d_{k,g,n,t}$ and $e_{g,n,t}$. 

\subsection{UAV's Energy Model}
We employ a UAV energy model proposed in \cite{zengenergyuav}.
The flying power is formulated as a function $\mathit{f}(U)$ of flying velocity $U$:
\begin{align}
\label{eq:power}
\mathit{f}(U) = &P_0\left(1+\frac{3U^2}{U_{tip}^2}\right)+P_1\left(\sqrt{1+\frac{U^4}{4U_{ind}^4}}-\frac{U^2}{2U_{ind}^2}\right)^{\frac{1}{2}}\notag \\
 &+\frac{1}{2}\rho_{1}\rho_{2}U^3,
\end{align}
where 
\begin{itemize}
\item $P_0$: the blade profile power in hovering status;
\item $P_1$: the induced power in hovering status;
\item $U_{tip}$: the tip speed of the rotor blade;
\item $U_{ind}$: the mean rotor induced velocity;
\item $\rho_1$: the parameter related to the fuselage drag ratio, rotor solidity, and the rotor disc area;
\item $\rho_2$: the air density.
\end{itemize}
When UAV approaches the hovering point of each cluster, it will fly around the point with
a certain velocity $U=U_{hov}$, which is more energy-efficient than $U=0$ \cite{Tranuav}.
Thus, the hovering power $P_H$ is $f(U=U_{hov})$.
The flying energy with constant velocity $U$ and traveling distance $S$ is expressed as:
\begin{align}
\label{eq:flyingenergy}
&\textit{f}(U)\cdot S/U \notag\\ 
=&~ SP_0\left(\frac{1}{U}+\frac{3U}{U_{tip}^2}\right)+SP_1\left(\sqrt{\frac{1}{U^4}+\frac{1}{4U_{ind}^4}}-\frac{1}{2U_{ind}^2}\right)^{\frac{1}{2}} \notag\\
 &+\frac{S}{2}\rho_{1}\rho_{2}U^2.
\end{align}
Hovering energy and communication energy need to be jointly optimized since they are coupled by hovering time, whereas the optimization of flying energy is independent.
By applying graph-based numerical methods \cite{flyinguav}, the minimum flying energy $E^*_{_F}$ along with the optimal flying speed $U^*_{_F}$ can be obtained by: 
\begin{align}
\label{eq:minflyenergy}
E_{_F}^*=\textit{f}(U_{_F}^*)\cdot S/U_{_F}^*,
\end{align}
where $U_{_F}^* = \argmin_{U\geq 0} \frac{\textit{f}(U)}{U}$.

The main notations are summarized in Table \ref{tab:notations}.
\begin{table}[h]
\footnotesize
\centering
\caption{Summary of Symbols and Notations}
\label{tab:notations}
\begin{tabular}{|c||c|}
 \hline
 Notation & Description\\
 \hline
 $N, \CN$ & number and set of clusters \\
 \hline
 $L$ & number of antennas in UAV\\
 \hline
 $K_n, \CK_n$ & number and set of users in cluster $n$\\ 
 \hline
 $G_n, \CG_n$ & number and set of groups in cluster $n$\\
 \hline
 $K_{g,n}, \CK_{g,n}$ & number and set of users in group $g$ of cluster $n$  \\ 
 \hline
 $q_{k,n}$ & demands of user $k$ in cluster $n$  \\ 
  \hline
 $T_{max}, \CT$ & maximum number and set of frames in each round\\
 \hline
 $I, \CI$ & number and set of timeslots in each frame   \\ 
 \hline
 $\Phi$ & duration of each timeslot (in seconds) \\
 \hline
 $\Gamma_{k,g,n,t}$ & SINR of user $k\in \CK_{g,n}$ on frame $t$\\ 
 \hline
 \multirow{2}*{$\beta_{g,n,t}^{(kj)}$} & channel coefficient from user $j$'s precoding\\
 ~& vector to user $k$ ($k,j \in \CK_{g,n}$) on frame $t$ \\ 
 \hline
 \multirow{2}*{$d_{k,g,n,t}$} &  transmitted data of user $k\in \CK_{g,n}$ per timeslot \\
 ~& on frame $t$\\
 \hline
 \multirow{2}*{$e_{g,n,t}$} &  communication energy of group $g\in \CG_{n}$ per \\
 ~& timeslot on frame $t$\\
 \hline
 \multirow{2}*{$U_{_F}^{*}$} & UAV's flying velocity that minimizes flying energy \\
 ~&with a predetermined flying path\\
 \hline
 \multirow{2}*{$E_{_F}^*$} & minimal flying energy with a predetermined \\
 ~&flying path\\
 \hline
\end{tabular}
\end{table}

\section{Problem Formulation}\label{sec:2.2}
We denote binary variables $\lambda_{i,g,n,t} \in \{0,1\}$ as the scheduling indicator, where $\lambda_{i,g,n,t}=1$ indicates that user group $g \in \CG_n$ is assigned to timeslot $i$ on frame $t$ and $\lambda_{i,g,n,t}=0$ otherwise.
Another  binary variables $\nu_{n,t} \in \{0,1\}$ indicate that the UAV is hovering above cluster $n$ on frame $t$ ($\nu_{n,t} = 1$), and $\nu_{n,t} = 0$ otherwise.
The UAV energy consumption consists of flying energy $E_{_F}$, hovering energy $E_{_H}$, and communication energy $E_{_C}$.
Since the minimal flying energy $E_{_F}^*$ can be independently obtained by Eq. (\ref{eq:minflyenergy}) without loss of optimality, the objective focuses on joint optimization of $E_{_C}$ and $E_{_H}$, which are expressed by:
\begin{align}
&E_{_C}=\sum_{t=1}^{T_{max}}\sum_{n=1}^{N}\sum_{g=1}^{G_n}\sum_{i=1}^{I}\nu_{n,t}\lambda_{i,g,n,t}e_{g,n,t}, \label{eq:commen}
\end{align}
\begin{align}
E_{_H}=\sum_{t=1}^{T_{max}}\sum_{n=1}^{N} \Phi I P_H \nu_{n,t}.
\end{align}
Note that the UAV is battery limited in practice. 
We focus on the instances that the minimum consumed energy in (\ref{OP:1}) is within the UAV's battery storage, otherwise the task is infeasible.
The optimization problem is formulated as:
\begin{subequations}
\begin{align}
&\mathcal{P}_1:~\min\limits_{\lambda_{i,g,n,t},\atop \nu_{n,t}} ~~ E_{_C}+E_{_H} \label{OP:1} \\
&s.t. \notag\\
& \sum_{t=1}^{_{T_{max}}}\sum_{g=1}^{_{G_n}}\sum_{i=1}^{_{I}}\nu_{n,t}\lambda_{i,g,n,t}d_{k,g,n,t} \geq q_{_{k,n}},~\forall k\in \CK_{n},~n\in \CN, \label{eq:demands} \\
& \nu_{_{n,t}} \leq \nu_{_{n,t\text{+1}}}+\nu_{_{n\text{+1},t\text{+1}}},~\forall n\in \CN,~t\in \CT, \label{eq:order1} \\
& \sum_{g=1}^{G_n}\sum_{i=1}^{I}\lambda_{i,g,n,t} = I\cdot \nu_{_{n,t}},~\forall n\in \CN^{+},~t\in \CT, \label{eq:nucons1} \\
& \sum_{g=1}^{G_n}\lambda_{i,g,n,t} \leq 1,~~\forall i \in \CI,~n\in \CN^{+},~t\in \CT,  \label{eq:varscons1}\\
& \sum_{n=1}^{N+1}\nu_{n,t}  = 1,~\forall t\in \CT,  \label{eq:varscons2} \\
& \lambda_{i,g,n,t} \in \{0,1\} ,~\forall i \in \CI,~g\in \CG_n,~n\in \CN^{+},~t \in \CT, \label{eq:varscons3}\\
& \nu_{n,t} \in \{0,1\},~\forall n\in \CN^{+},~t \in \CT. \label{eq:varscons4}
\end{align}
\end{subequations}
Constraints (\ref{eq:demands}) guarantee that all the users' requests have to be satisfied within $T_{max}$.
Constraints (\ref{eq:order1}) define that the UAV follows a successive and forward manner in visiting clusters.  
For example, if the UAV is hovering above cluster $n$ on frame $t$, in the next frame $t+1$, the UAV either chooses to stay at the current cluster $n$ or move to the next cluster $n+1$.
The option of flying back to previously visited clusters, e.g., $n-1$, is thus excluded. 
Note that the UAV takes off from the first cluster, i.e., $\nu_{_{1,1}}=1$.
Constraints (\ref{eq:nucons1}) represent that all the timeslots on frame $t$ are assigned to a user group when $\nu_{n,t} = 1$, otherwise, no users are scheduled in any timeslot.
Constraints (\ref{eq:varscons1}) and (\ref{eq:varscons2}) indicate that no more than one group can be scheduled at a timeslot and only one cluster can be served within a frame.
Constraints (\ref{eq:varscons3}) and (\ref{eq:varscons4}) confine variables $\lambda_{i,g,n,t}$ and $\nu_{n,t}$ to binary. 

Note that $\mathcal{P}_1$ is a combinatorial optimization problem with a non-convex bilinear objective and constraints.
The optimum can be approached by a well-established relax-and-approximate method. 
That is, the non-convex bilinear terms are relaxed and bounded by McCormick envelop \cite{McCormick}, where each variable ($\lambda_{i,g,n,t}$ and $\nu_{n,t}$) is bounded by an upper and a lower bound.
The relaxation problem becomes an integer linear programming (ILP) problem which can be optimally solved by B\&B. 
Overall, the optimum of $\mathcal{P}_1$ can be approached by ultimately tightening the bounds, e.g., increase the number of breakpoints in the envelopes, but this results in exponentially increasing complexity which is unaffordable in practice \cite{wangbb}.
Thus, we adopt the above relax-and-approximate method to provide an optimal solution for benchmarking small-medium cases.
For general cases, we propose a sub-optimal algorithm in the next section. 

\section{Heuristic Approach}\label{sec:3}
We decompose the joint optimization to two sub-problems, i.e., user-timeslot and hovering time allocation, corresponding to optimization of $\lambda_{i,g,t,n}$ and $\nu_{n,t}$, respectively. 
We then solve one sub-problem when the other is fixed.

\subsection{User-Timeslot Scheduling}
The bilinear items are resolved with the fixed $\nu_{n,t}$.
The number of frames at each cluster are determined by: 
\begin{align}
t_n = \sum_{t=1}^{T_{max}}\nu_{n,t},~\forall n\in \CN,
\end{align}
and $\Phi I t_n$ is the hovering duration.
The user-timeslot scheduling can be carried out independently in each cluster, and the resulting problem for the $n$-th cluster is formulated in $\mathcal{P}_2(n)$ with a given $t_n$.
We denote $E_{_{H,n}}$ and $E_{_{C,n}}$ as the hovering and communication energy for the $n$-th cluster:
\begin{align}
&E_{_{H,n}}= \Phi I P_{_H} t_n,\label{eq:EHn}\\
&E_{_{C,n}}=\sum_{t=\tau_n+1}^{\tau_n+t_n}\sum_{g=1}^{G_n}\sum_{i=1}^{I}\lambda_{i,g,n,t}e_{g,n,t}, \label{eq:ECn}
\end{align}
where $\tau_n$ refers to the number of elapsed frames before the UAV arriving cluster $n$, which can be calculated by:
\begin{align}
\tau_n=\sum_{t=1}^{T_{max}}\sum_{n'=1}^{n-1}\nu_{n',t}.
\end{align}
The sub-problem $\mathcal{P}_2(n)$ is formulated as:
\begin{subequations}
\begin{align}
\label{OP:2}
&\mathcal{P}_2(n):~\min\limits_{\lambda_{i,g,n,t}} ~~ E_{_{C,n}}+E_{_{H,n}} \\
&s.t. \notag\\
& \sum_{t=\tau_n+1}^{\tau_n+t_n}\sum_{g=1}^{G_n}\sum_{i=1}^{I}\lambda_{i,g,n,t}d_{k,g,n,t} \geq q_{k,n},~ \forall k\in \CK_{n}, \label{eq:demands2} \\
& \sum_{g=1}^{G_n}\sum_{i=1}^{I}\lambda_{i,g,n,t} = I,~\forall t\in\{\tau_n+1,...,\tau_n+t_n\}, \label{eq:nucons1_2}\\
& \sum_{g=1}^{G_n}\lambda_{i,g,n,t} \leq 1,~~\forall i \in \CI,~t\in \CT,  \label{eq:varscons1_2}\\
& \lambda_{i,g,n,t} \in \{0,1\} ,~\forall i \in \CI,~g\in \CG_n,~t \in \CT.  \label{eq:varscons3_2}
\end{align}
\end{subequations}
$\mathcal{P}_2(n)$ is a multi-choice multi-dimensional knapsack problem (MMKP), which can be solved by a guided local search (GLS)-based heuristic algorithm with high-quality sub-optimal solutions and pseudo-polynomial-time complexity \cite{MMKP}.

\subsection{Hovering Time Allocation}
To optimize hovering time efficiently, we first investigate the connection between the objective energy and $t_n$.
From Eq. (\ref{eq:EHn}) and Eq. (\ref{eq:ECn}), $E_{_{H,n}}$ increases linearly with $t_n$ while $E_{_{C,n}}$ is determined by both $t_n$ and $\lambda_{i,g,n,t}$.
Next, we show the relationship between the optimum $E_{_{C,n}}$ and $t_n$.
For cluster $n$, we denote $E_{_{C,n}}^{*}(t_n)$ as the communication energy with the optimal scheduling decision $\lambda^{*}_{i,g,n,t}$ at a given hovering time $t_n$. 
\begin{lemma}
\label{le:1} 
$E_{_{C,n}}^{*}(t_n)$ is a non-increasing function of $t_n$,
\begin{align}
E_{_{C,n}}^*(\hat{t}) \geq E_{_{C,n}}^*(\hat{t}+\Delta t),~\hat{t}>0,\Delta t>0.
\end{align}
\end{lemma}
\begin{proof}
We denote the optimal user scheduling for $\mathcal{P}_2(n)|_{t_n=\hat{t}}$ as $\lambda_{i,g,n,t}^*$.
If $t_n$ increases from $\hat{t}$ to $\hat{t}+\Delta t$, $\lambda_{i,g,n,t}^*$ is still feasible for $\mathcal{P}_2(n)|_{t_n=\hat{t}+\Delta t}$ such that
\begin{align}
\label{eq:lemma1_1}
&E_{_{C,n}}^*(\hat{t}) = \sum_{t=\tau_n+1}^{\tau_n+\hat{t}}\sum_{g=1}^{G_n}\sum_{i=1}^{I}\lambda_{i,g,n,t}^*e_{g,n,t}\notag\\
=&E_{_{C,n}}^{'}(\hat{t}+\Delta t) = \sum_{t=\tau_n+1}^{\tau_n+\hat{t}+\Delta t}\sum_{g=1}^{G_n}\sum_{i=1}^{I}\lambda_{i,g,n,t}^*e_{g,n,t}.
\end{align}
$\lambda_{i,g,n,t}^{*}$ might not be necessarily optimal for $t_n = \hat{t}+\Delta t$.
There exists an optimal scheduling resulting in lower communication energy, i.e.,
\begin{align}
\label{eq:lemma1_2}
&E_{_{C,n}}^*(\hat{t}+\Delta t)\leq E_{_{C,n}}^{'}(\hat{t}+\Delta t)=E_{_{C,n}}^*(\hat{t}).
\end{align} 
Thus the conclusion.
\end{proof}

From Lemma \ref{le:1}, we can observe that $E_{_{C,n}}^*(t_n)$ is an non-increasing function of $t_n$, i.e., $\frac{dE_{_{C,n}}^*(t_n)}{dt_n}\leq 0$.
For $E_{_{H,n}}(t_n)$, we can derive that $\frac{dE_{_{H,n}}(t_n)}{dt_n}=\Phi I P_H$ based on Eq. (\ref{eq:EHn}).
Thus, the extreme point of $E_{_{C,n}}^*(t_n)+E_{_{H,n}}(t_n)$ can be obtained at $t_n = t^{\dagger}$  when
\begin{align}
\label{eq:extreme_point}
\frac{dE_{_{C,n}}^*(t_n)}{dt_n}|_{t_n = t^{\dagger}}=-\Phi I P_{_H}.
\end{align}
Since the existence and the number of extreme points are undetermined.
There are three possible cases, i.e., unimodal, multimodal, and monotonic, for $E_{_{C,n}}^*(t_n)+E_{_{H,n}}(t_n)$, as illustrated in Fig. \ref{fig:lemma2}. 
In case 1, the curve is a unimodal function with only one extreme point.
In case 2, the fluctuation of $\frac{dE_{_{C,n}}^*(t_n)}{dt_n}$ leads to multiple extreme points such that the curve is a multimodal function.
In case 3, Eq. (\ref{eq:extreme_point}) cannot hold, e.g., $\frac{dE_{_{C,n}}^*(t_n)}{dt_n}$ is consistently lager than $-\Phi I P_{_H}$, so the curve is monotonously increasing with no extreme point.
\begin{figure}[h]
\begin{center}
\centering
\vspace{-0.2cm}  %调整图片与上文的垂直距离
\includegraphics[scale=0.65]{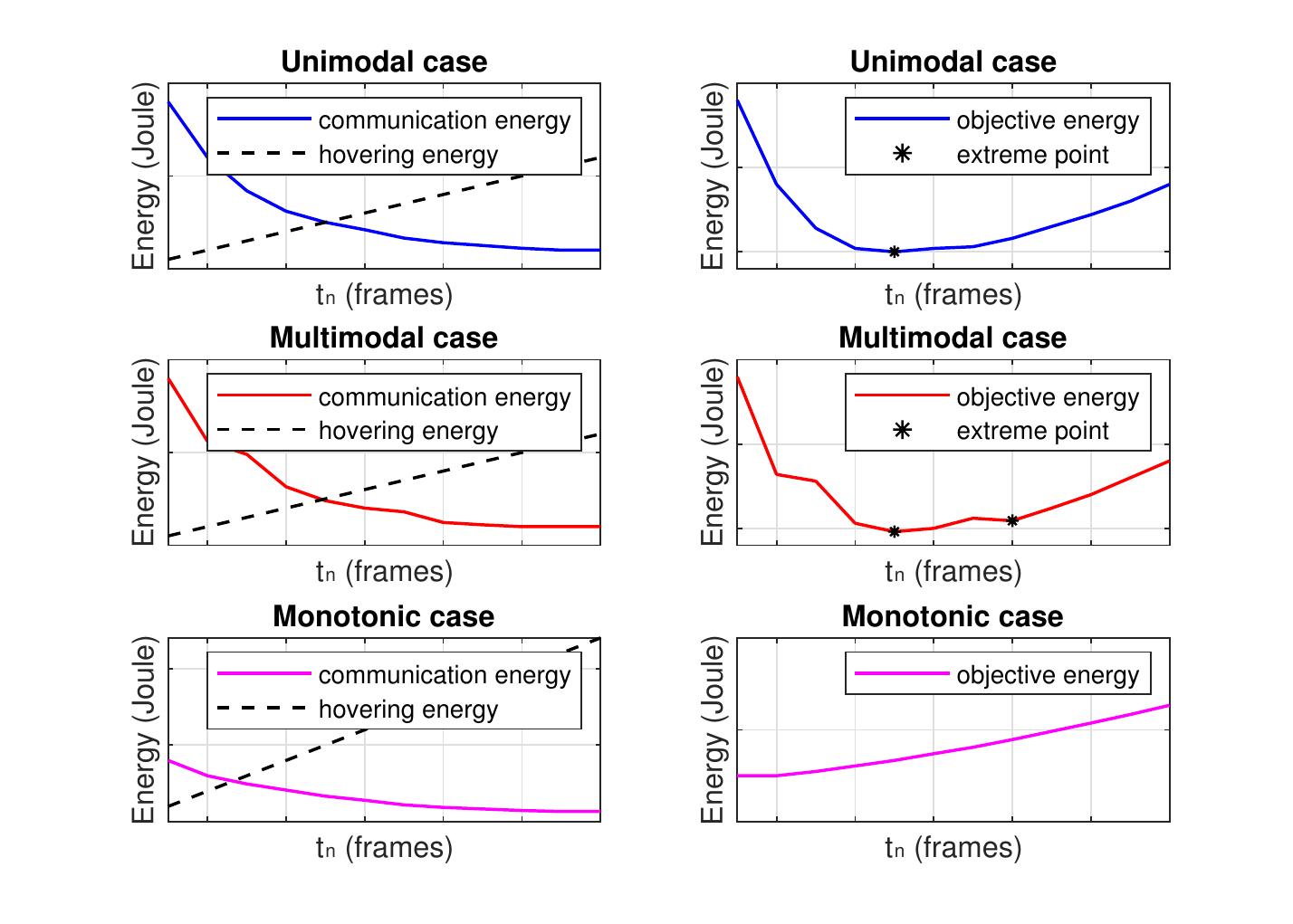}
\setlength{\abovecaptionskip}{-0.4cm}   %调整图片标题与图距离
\setlength{\belowcaptionskip}{-0.35cm}   %调整图片标题与下文距离
\captionsetup{font={small}}
\caption{Energy curves for three possible cases.}
\label{fig:lemma2}
\end{center}
\end{figure}

Observing the possible cases, we employ an efficient golden section search (GSS) to find the extreme points \cite{guigolden}.
In GSS, we limit the hovering time $t_n\leq \bar{t}_{n}$ to ensure that the total service duration does not exceed $T_{max}$, where $\bar{t}_{n}$ is a maximal time limitation for cluster $n$.
Intuitively, the clusters with more demands need more transmission frames.
We assume $\bar{t}_{n}$ is proportional to the users' demands:
\begin{equation}
\bar{t}_{n} = T_{max}\frac{\sum_{k=1}^{K_n}q_{k,n}}{\sum_{n=1}^{N}\sum_{k=1}^{K_n}q_{k,n}}.
\end{equation} 

\subsection{Algorithm Summary}
We summarize the proposed GSS-based heuristic (GSS-HEU) algorithm in Alg. \ref{alg:Multiple}.
We denote $\CB_{n,t}$ as the set of channel states of cluster $n$ on frame $t$, which is expressed as:
\begin{align}
\CB_{n,t}=\{\beta^{_{(kj)}}_{_{1,n,t}},...,\beta^{_{(kj)}}_{_{G_n,n,t}}|~\forall k,j \in \CK_{g,n}\}.
\end{align}
In GSS-HEU, the initial search range of GSS $[x_1,y_1]$ is set as $[0, \bar{t}_n]$, which is partitioned into 3 sections by two points $u_1$ and $v_1$ with the golden ratio 0.618 in lines 2-4, where $\lceil\centerdot\rceil$ is an operation to round a value up to an integer.
When a hovering time is searched in GSS, e.g., $t_n=u_m$ or  $t_n=v_m$, the corresponding user-timeslot allocation is obtained by solving $\mathcal{P}_2(n)$ in line 6.
In lines 9-13, we compare the objective energy and update the search range.
The search process terminates at $|y_{m}-x_{m}|\leq 1$.
The selected hovering time $t_n^*$ is $v_m$ and the corresponding scheduling scheme $\lambda_{{i,g,n,t}}^*$ is $\lambda_{i,g,n,t}|_{t_n=v_m}$.

\begin{algorithm}[htb]
  \caption{GSS-HEU Algorithm}
  \label{alg:Multiple}
  \begin{algorithmic}[1]
  	\REQUIRE ~~\\
  	Users' demands: $q_{_{1,1}}$,$...$, $q_{_{K_1,1}}$,$...$, $q_{_{1,N}}$,$...$,$q_{_{K_N,N}}$;\\
	Channel states: $\CB_{^{1,1}}$,$...$,$\CB_{^{1,T_{max}}}$,$...$,$\CB_{^{N,1}}$,$...$,$\CB_{^{N,T_{max}}}$;\\
	Search range's upper bound: $\bar{t}_{_1},...,\bar{t}_{_N}$.
  	\ENSURE ~~\\ 
  	Heuristic solution: $\lambda_{^{1,1,1,1}}^*,...,\lambda_{^{I,G_n,N,T_{max}}}^*,t_{^1}^{*},...,t_{^N}^*$
  	\FOR {$n=1$; $n\leq N$; $n++$}
  	\STATE $x_1 = 0$; $y_1 = \bar{t}_{n}$;
  	\STATE $u_1 = \lceil y_1-0.618(y_1-x_1)\rceil$;\\
  	\STATE $v_1 = \lceil x_1+0.618(y_1-x_1) \rceil$;
			\FOR {$m=1$; $|y_m-x_m|> 1$; $m++$}
			\STATE Solve $\mathcal{P}_2(n)|_{t_n=u_m}$ and $\mathcal{P}_2(n)|_{t_n=v_m}$;
			\STATE Obtain the corresponding user scheduling schemes $\lambda_{i,g,n,t}|_{t_n=u_m}$ and $\lambda_{i,g,n,t}|_{t_n=v_m}$;
			\STATE Obtain the objective energy $(E_{_{C,n}}+E_{_{H,n}})|_{t_n=u_m}$ and $(E_{_{C,n}}+E_{_{H,n}})|_{t_n=v_m}$;
			\IF {$(E_{_{C,n}}+E_{_{H,n}})|_{t_n=u_m}<(E_{_{C,n}}+E_{_{H,n}})|_{t_n=v_m}$}
				\STATE $x_{m+1}=x_m;~y_{m+1}=v_m;~v_{m+1}=u_m$;\\
				$u_{m+1} =\lceil y_{m+1}-0.618(y_{m+1}-x_{m+1})\rceil$;
			\ELSE
			    \STATE $x_{m+1}=u_m;~y_{m+1}=y_m;~u_{m+1}=v_m$;\\
			    $v_{m+1} =\lceil y_{m+1}-0.618(y_{m+1}-x_{m+1})\rceil$;
			\ENDIF
        	\ENDFOR
			\STATE $t_n^*=v_m$; $\lambda_{{i,g,n,t}}^*=\lambda_{i,g,n,t}|_{t_n=v_m}$.
  	\ENDFOR
     \end{algorithmic}
\end{algorithm}

The complexity of GSS-HEU is $\mathcal{O}(\sum_{n=1}^N G_{n}^2\times \text{max}\{K_{n},I\bar{t}_{n}\}+\text{log}(2\bar{t}_{n}))$, which is much lower than that of the optimal method.
However, both the optimal and GSS-HEU approaches may have limitations in fast decision-making.
The computational time for both algorithms grows exponentially with the number of users since $G_n = 2^{K_n}-1$ \cite{leischeduling}.
In addition, both algorithms need the estimated and complete channel states for the whole task frames, i.e., from $t=1$ to $T_{max}$. 
This may result in difficulties in channel estimation.
Therefore, we reconsider $\mathcal{P}_1$ from the perspective of DRL to enable the UAV to make decisions intelligently, while the developed optimal and sub-optimal algorithms are used to benchmark the performance of learning-based solutions.

\section{Actor-Critic-Based DRL algorithm}\label{sec:4}
\subsection{Overview of Actor-Cirtic-Based DRL (AC-DRL)}
In DRL, an agent learns to make decisions by exploring the unknown environments and exploiting the received feedbacks.
At each learning step\footnote{In this paper, a learning step is equivalent to a transmission frame.} $t$, 
the agent observes the current state $\boldsymbol{s}_t$ and takes an action $\boldsymbol{a}_t$ based on a policy.
Then, a reward $r_t$ will be fed back to the agent.
The policy will be updated step by step according to the feedback.
Actor-critic is an emerging reinforcement learning method that separates the agent into two parts, an actor and a critic.
The actor is responsible for taking actions following a stochastic policy $\pi(\boldsymbol{a}_t|\boldsymbol{s}_t)$, where $\pi(\centerdot|\centerdot)$ refers to a conditional probability density function.
The critic is used to evaluate the decisions via a Q-value, which is given by:
\begin{align}
Q^{\pi}(\boldsymbol{s}_t,\boldsymbol{a}_t)=\mathbb{E}_{\boldsymbol{a}_t\sim\pi(\boldsymbol{a}_t|\boldsymbol{s}_t)}[R_t|\boldsymbol{s}_t, \boldsymbol{a}_t],
\end{align}
where $\mathbb{E}_{\boldsymbol{a}_t\sim\pi(\boldsymbol{a}_t|\boldsymbol{s}_t)}[\centerdot|\centerdot]$ is a conditional expectation under the policy $\pi(\boldsymbol{a}_t|\boldsymbol{s}_t)$, and $R_t$ is the cumulative discounted reward with a discount factor $\gamma$, which can be expressed as:
\begin{align}
R_t = \sum_{t'=t}^{\infty}\gamma^{t'-t}r_{t'},~~\gamma\in [0,1].
\end{align}
However, obtaining the explicit expressions of $\pi(\boldsymbol{a}_t|\boldsymbol{s}_t)$ and $Q^{\pi}(\boldsymbol{s}_t,\boldsymbol{a}_t)$ is difficult.
DRL uses DNNs as the parameterized approximators to provide estimations for $\pi(\boldsymbol{a}_t|\boldsymbol{s}_t)$ and $Q^{\pi}(\boldsymbol{s}_t,\boldsymbol{a}_t)$.
We denote $\boldsymbol{\theta}_t$ and $\boldsymbol{\omega}_t$ as the parameter vectors for the actor and critic, and $\pi(\boldsymbol{a}_t|\boldsymbol{s}_t;\boldsymbol{\theta}_t)$ and $Q^{\boldsymbol{\theta}}(\boldsymbol{s}_t,\boldsymbol{a}_t;\boldsymbol{\omega}_t)$ as the corresponding parameterized functions\footnote{ For simplicity, $Q^{\boldsymbol{\theta}}(\boldsymbol{s}_t,\boldsymbol{a}_t;\boldsymbol{\omega}_t)=\mathbb{E}_{\boldsymbol{a}_t\sim\pi(\boldsymbol{a}_t|\boldsymbol{s}_t;\boldsymbol{\theta}_t)}[R_t|\boldsymbol{s}_t, \boldsymbol{a}_t]$.}.
The goal of the agent is to minimize the loss function of the actor $-J(\boldsymbol{\theta}_t)$:
\begin{align}
\label{eq:loss_actor}
-J(\boldsymbol{\theta}_t) = -\mathbb{E}[Q^{\boldsymbol{\theta}}(\boldsymbol{s}_t,\boldsymbol{a}_t;\boldsymbol{\omega}_t)].
\end{align}  
Based on the fundamental results of the policy gradient theorem \cite{IntroRL}, the gradient of $J(\boldsymbol{\theta}_t)$ can be calculated by:
\begin{align}
\label{eq:policy_gradient}
\nabla_{\boldsymbol{\theta}}J(\boldsymbol{\theta}_t) = \mathbb{E}[\nabla_{\boldsymbol{\theta}}\log\pi(\boldsymbol{a}_t|\boldsymbol{s}_t;\boldsymbol{\theta}_t)Q^{\boldsymbol{\theta}}(\boldsymbol{s}_t,\boldsymbol{a}_t;\boldsymbol{\omega}_t)].
\end{align}
The update rule of $\boldsymbol{\theta}_t$ can be derived based on gradient descent:
\begin{align}
\label{eq:update_actor}
\boldsymbol{\theta}_{t+1} =  \boldsymbol{\theta}_{t} - \alpha_a\cdot(-\nabla_{\boldsymbol{\theta}} J(\boldsymbol{\theta}_{t})),
\end{align}
where $\alpha_a$ is the learning rate of the actor.
For the critic, the parameter vector $\boldsymbol{\omega}_t$ is updated based on temporal-difference (TD) learning \cite{IntroRL}.
In TD learning, the loss function of the critic $C_{_Q}(\boldsymbol{\omega}_t)$ is defined as the expectation of the square of TD error $\delta_{_Q}(\boldsymbol{\omega}_t)$, i.e., $\mathbb{E}[(\delta_{_Q}(\boldsymbol{\omega}_t))^2]$.
The TD error $\delta_{_Q}(\boldsymbol{\omega}_t)$ refers to the difference between the TD target and estimated Q-value, which is given by:
\begin{align}
\label{eq:TD_Q}
\delta_{_Q}(\boldsymbol{\omega}_t) = r_t+\gamma Q^{\boldsymbol{\theta}}(\boldsymbol{s}_{t+1},\boldsymbol{a}_{t+1};\boldsymbol{\omega}_t)-Q^{\boldsymbol{\theta}}(\boldsymbol{s}_t,\boldsymbol{a}_t;\boldsymbol{\omega}_t),
\end{align}
where $r_t+\gamma Q^{\boldsymbol{\theta}}(\boldsymbol{s}_{t+1},\boldsymbol{a}_{t+1};\boldsymbol{\omega}_t)$ is the TD target.
The objective of the critic is to minimize the loss function $C_{_Q}(\boldsymbol{\omega}_t)$ and the updated rule of $\boldsymbol{\omega}_t$ can be derived by gradient descent:
\begin{align}
\label{eq:updatew_critic}
\boldsymbol{\omega}_{t+1} = \boldsymbol{\omega}_t - \alpha_c\nabla_{\boldsymbol{\omega}} C_{_Q}(\boldsymbol{\omega}_t),
\end{align}
where $\alpha_c$ is the learning rate for the critic.

However, approximating $Q^\pi(\boldsymbol{s}_t,\boldsymbol{a}_t)$ brings about a large variance for the gradient $\nabla_{\boldsymbol{\theta}}J(\boldsymbol{\theta}_t)$, resulting in poor convergence \cite{advantagefunc}.
To solve the problem, a V-value is introduced:
\begin{align}
V^\pi(\boldsymbol{s}_t) =\mathbb{E}_{\boldsymbol{a}_t\sim\pi(\boldsymbol{a}_t|\boldsymbol{s}_t)}[R_t|\boldsymbol{s}_t].
\end{align}
Approximating $V^\pi(\boldsymbol{s}_t)$ can reduce the variance.
With the parametered V-value $V^{\boldsymbol{\theta}}(\boldsymbol{s}_t; \boldsymbol{\omega}_t)$, the TD error and the loss function of the critic are expressed as:
\begin{align}
\label{eq:TD_V}
\delta_{_V}(\boldsymbol{\omega}_t) = r_t+\gamma V^{\boldsymbol{\theta}}(\boldsymbol{s}_{t+1};\boldsymbol{\omega}_t)-V^{\boldsymbol{\theta}}(\boldsymbol{s}_t;\boldsymbol{\omega}_t),
\end{align}
and 
\begin{align}
\label{eq:loss_TD_V}
C_{_V}(\boldsymbol{\omega}_t) = \mathbb{E}[(\delta_{_V}(\boldsymbol{\omega}_t))^2].
\end{align}
In addition, $\delta_{_V}(\boldsymbol{\omega}_t)$ provides an unbiased estimation of Q-value \cite{advantagefunc}. 
Thus, we can rewrite $\nabla_{\boldsymbol{\theta}} J(\boldsymbol{\theta}_t)$ in Eq. (\ref{eq:policy_gradient}) as:
\begin{align}
\nabla_{\boldsymbol{\theta}} J(\boldsymbol{\theta}_t) =& \mathbb{E}\left[\nabla_{\boldsymbol{\theta}}\log(\pi(\boldsymbol{a}_t|\boldsymbol{s}_t;\boldsymbol{\theta}_t))Q^{\pi}(\boldsymbol{s}_t,\boldsymbol{a}_t)\right]\notag\\
=&\mathbb{E}\left[\nabla_{\boldsymbol{\theta}}\log(\pi(\boldsymbol{a}_t|\boldsymbol{s}_t;\boldsymbol{\theta}_t))\delta_{_V}(\boldsymbol{\omega}_t)\right].
\end{align}

\subsection{Problem Reformulation}
To apply AC-DRL, we reformulate $\mathcal{P}_1$ to an MDP problem, in which the UAV acts as an agent.
We define the states, actions, and rewards as follows.

\subsubsection{States}
The system states $\boldsymbol{s}_t$ consist of the channel states for all
the clusters on the current frame, i.e., $\CB_{1,t},...,\CB_{N,t}$, the undelivered demands, and the currently served cluster on frame $t$. 
The undelivered demands $b_{n,t}$ is the residual data to be delivered for cluster $n$ on frame $t$:
\begin{align}
\label{eq:bt}
&b_{n,t+1} = b_{n,t} - d^{\pi}_{n,t},~\forall n\in\CN,~t\in\CT,\\
&b_{n,0} = \sum_{k=1}^{K_n}q_{k,n},~\forall n\in\CN,
\end{align}
where $d^{\pi}_{n,t}$ is the delivered data for cluster $n$ in frame $t$ under the policy $\pi(\boldsymbol{s}_t|\boldsymbol{a}_t)$.
We denote $o_t \in \CN^{+}$ as an indicator to represent which cluster the UAV is serving in frame $t$.
%Based on the undelivered demands $d_{n,t}$, $o_{t+1}$ can be determined automatically.
When the users’ requests in the current cluster are completed, the UAV will move to the next cluster in the next frame, otherwise, staying at the current cluster.
For example, we assume that the UAV is hovering above cluster $n$ on frame $t$, i.e., $o_{t}=n$.
For the next frame, $o_{t+1}$ is obtained by: 
\begin{align}
\label{eq:ot}
o_{t+1} =\left\lbrace
\begin{array}{ll}
n, & b_{n,t} > 0,\\
n+1, & b_{n,t} = 0.\\
\end{array}
\right.
\end{align}
When the UAV's duration exceeds $T_{max}$, the UAV will fly back to the dock station.
By assembling the above three parts, the state $\boldsymbol{s}_t$ is defined as:
\begin{align}
\boldsymbol{s}_t=[\CB_{_{1,t}},...,\CB_{_{N,t}},b_{_{1,t}}...,b_{_{N,t}}, o_t].
\end{align}
Note that the elements of $\CB_{n,t}$ are modeled as FSMC.
In addition, based on Eq. (\ref{eq:bt}) and Eq. (\ref{eq:ot}), the next state of $b_{n,t}$ and $o_{t}$ only depend on the current state and current policy.
Therefore, the transition of the state $\boldsymbol{s}_t$ conforms to MDP \cite{IntroRL}.
\subsubsection{Actions}	
The action of the UAV is the user-timeslot assignment on frame t, which is given by:
\begin{align}
&\boldsymbol{a}_t=[a_{_{1,t}},...,a_{_{I,t}}],\notag\\
&a_{_{i,t}} \in \{1,...,g,...,G_n\},~\forall i \in \CI,~t \in \CT,
\end{align}
where $a_{i,t}=g$ means the $g$-th group is selected at the $i$-th timeslot on the $t$-th frame.
Note that the action space $G_n$ can be huge since it increases exponentially with the number of users.

\subsubsection{Rewards}\label{sss:reward_design}
The reward functions are commonly related to the objective of the problem. Conventionally, the reward function of $\mathcal{P}_1$ can be designed by Eq. (\ref{eq:reward1}) and Eq. (\ref{eq:reward2}), referring to \cite{Wanguavconstraint} and \cite{Luddpg}:
\begin{equation}
\label{eq:reward1}         
      r_t=1/e^{\pi}_t,
\end{equation} 
\begin{equation}
\label{eq:reward2}         
      r_t=-e^{\pi}_t,
\end{equation} 
where $e^{\pi}_t$ is the energy consumed on frame $t$ under the policy $\pi(\boldsymbol{s}_t|\boldsymbol{a}_t)$. 
Since both the above reward functions monotonically decrease with $e^{\pi}_t$, the UAV updates the policy towards reducing energy consumption.

\subsection{The AC-DSOS algorithm}\label{sec:5}
Conventional AC-DRL algorithms may not be able to deal with constrained discrete problems.
Firstly, the combinatorial component of $\mathcal{P}_1$ limits the conventional AC-DRL in addressing huge discrete action spaces \cite{drlhuge}.
Secondly, the increased action space reduces the exploration efficiency in the learning process and degrades overall energy-saving performance.
Thirdly, the conventional AC-DRL algorithms cannot guarantee the solution's feasibility in general.
This means that a high-reward action can fail to satisfy the constraints in $\mathcal{P}_1$.
To overcome the above difficulties and limitations, we propose an AC-DSOS algorithm that is tailored for constrained problems with discrete action representation.
The basic actor-critic framework is employed in order to take the advantages of the stochastic policy and TD learning, where the stochastic policy can be quantified to tackle the issue of huge discrete spaces and TD learning can improve the learning efficiency.

\begin{figure*}
\begin{center}
\centering
\vspace{-0.2cm}  %调整图片与上文的垂直距离
\includegraphics[scale=0.6]{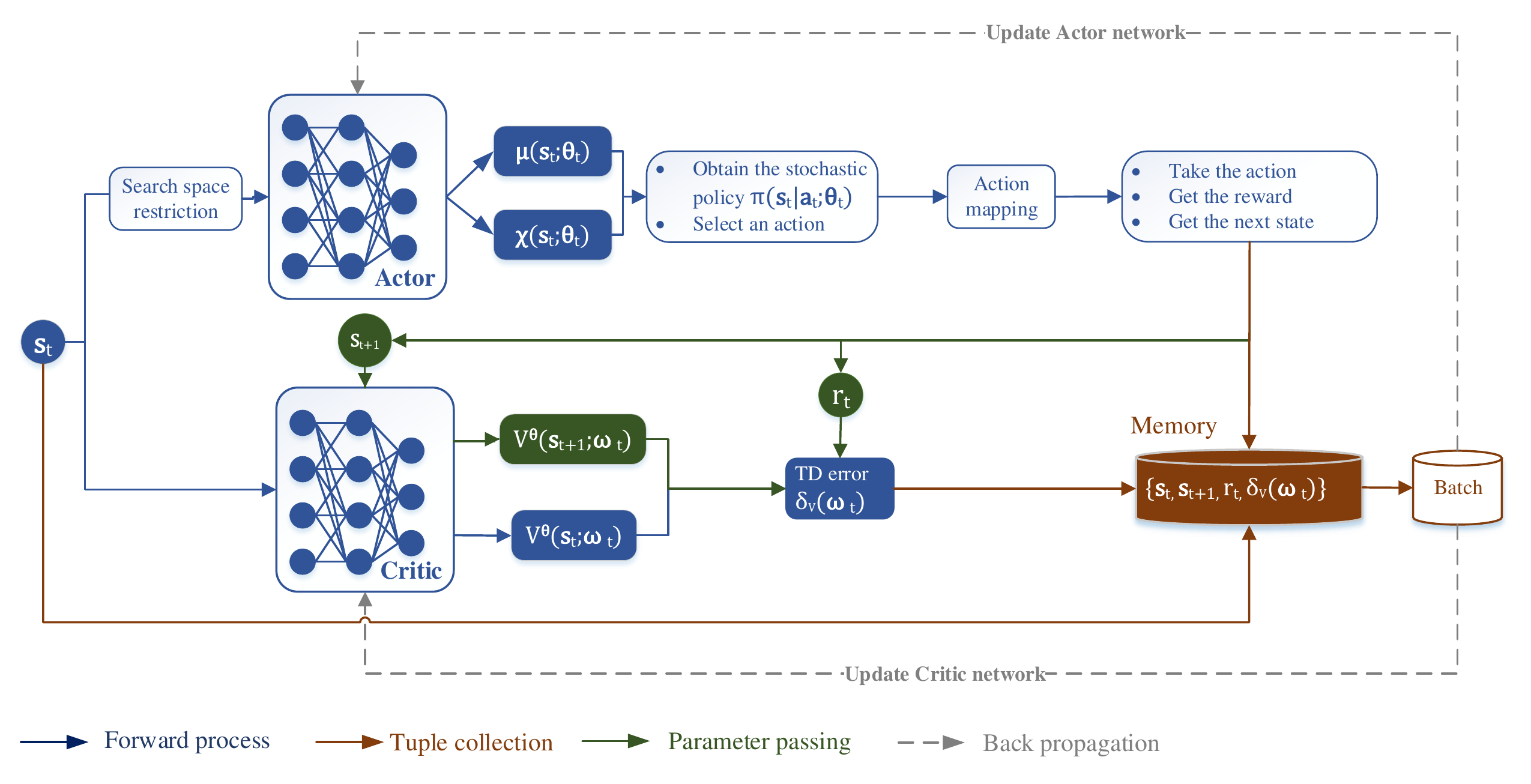}
\setlength{\belowcaptionskip}{-0.2cm}   %调整图片标题与下文距离
\captionsetup{font={small}}
\caption{The actor-critic framework of AC-DSOS.}
\label{fig:structure}
\end{center}
\end{figure*}	

We illustrate the actor-critic framework of AC-DSOS in Fig. \ref{fig:structure}, where two DNNs work as the actor and critic, respectively.
The stochastic policy $\pi(\boldsymbol{a}_t|\boldsymbol{s}_t)$ is usually modeled as Gaussian distribution with a mean $\boldsymbol{\mu}(\boldsymbol{s}_t)$ and a variance $\boldsymbol{\chi}(\boldsymbol{s}_t)$ \cite{weiuserschedule}.
Given the current state $\boldsymbol{s}_t$, the actor does not predict $\pi(\boldsymbol{a}_t|\boldsymbol{s}_t; \boldsymbol{\theta}_t)$ directly but obtains approximations of the mean $\boldsymbol{\mu}(\boldsymbol{s}_t;\boldsymbol{\theta}_t)$ and the variance $\boldsymbol{\chi}(\boldsymbol{s}_t;\boldsymbol{\theta}_t)$.
An action $\boldsymbol{a}_t$ can be selected based on $\pi(\boldsymbol{a}_t|\boldsymbol{s}_t; \boldsymbol{\theta}_t)$.
Then, the agent receives a reward $r_t$ after taking the action and collects the next state $\boldsymbol{s}_{t+1}$.
For the critic, two V-values, $V^{\boldsymbol{\theta}}(\boldsymbol{s}_t;\boldsymbol{\omega}_t)$ and $V^{\boldsymbol{\theta}}(\boldsymbol{s}_{s+1};\boldsymbol{\omega}_t)$, are estimated by DNN with the inputs $\boldsymbol{s}_t$ and $\boldsymbol{s}_{t+1}$, respectively.   
The TD error $\delta_{_V}(\boldsymbol{\omega}_t)$ can be calculated by Eq. (\ref{eq:TD_V}).
A tuple $\{\boldsymbol{s}_t,\boldsymbol{s}_{t+1},\delta_{_V}(\boldsymbol{\omega}_t),r_t\}$ is stored in a memory at each step $t$.
By applying a memory replay mechanism, the data in the memory can be used for training the DNNs.
In each training step, the actor and critic are updated by the gradient descent over a batch of training data.
The whole training process consists of multiple episodes, each episode including $T_{max}$ steps.
Based on the above framework, the AC-DSOS algorithm is summarized in Alg. \ref{alg:AC}.
The novelties of the proposed AC-DSOS compared to the conventional AC-DRL are summarized as follows.
\subsubsection{Action Mapping to Tackle the Issue of Huge Discrete Action Space}
The conventional actor-critic is used for continuous action space.
We denote $\hat{\boldsymbol{a}}_{t}=[\hat{a}_{1,t},...,\hat{a}_{I,t}]$ as the original action selected by the stochastic policy, where the element $\hat{a}_{i,t}$ is fractional.
However, as the decision variables are integers in $\mathcal{P}_1$, the action space is discrete.
To deal with this issue, we adopt an action mapping method in AC-DSOS (line 9 in Alg. \ref{alg:AC}).
Firstly, we confine $\hat{a}_{i,t}$ to a fixed range $[-\kappa,\kappa]$ to avoid its value being too large/small since the domain of Gaussian distribution is $[-\infty, \infty]$.
Then, a uniform quantization method is used to map $\hat{a}_{i,t}$ to the discrete action space $\{1,...,G_n\}$ by:
\begin{align}
\label{Eq:quantify} 
      a_{i,t} =  \lceil \frac{\kappa + \hat{a}_{i,t}}{2\kappa/{G_n}}\rceil,
\end{align} 
where $2\kappa/G_n$ is the quantization interval.
With the mapping operation, we can support a larger $G_n$ by reducing the interval.

\begin{algorithm}[h]
  \caption{AC-DSOS Algorithm}
  \label{alg:AC}
  \begin{algorithmic}[1]
  \REQUIRE The current state $\boldsymbol{s}_t$.
  \ENSURE The current action $\boldsymbol{a}_t$.
   \STATE Initialize $\boldsymbol{\theta}_1$ and $\boldsymbol{\omega}_1$.
	\FOR {each learning episode}
		\STATE Observe the initial state $\boldsymbol{s}_1$.
	\FOR {$t=1:T_{max}$}
			\STATE Remove the groups containing the demand-satisfied users.
			\STATE Predicted mean $\boldsymbol{\mu}(\boldsymbol{s}_t;\boldsymbol{\theta}_t)$ and variance $\boldsymbol{\chi}(\boldsymbol{s}_t;\boldsymbol{\theta}_t)$ by the DNN of the actor.	
			\STATE Obtain action's distribution $\pi(\boldsymbol{a}_t|\boldsymbol{s}_t;\boldsymbol{\theta}_t)$ based on Gaussian distribution.
			\STATE Randomly choose $\hat{\boldsymbol{a}}_t$ following $\pi(\boldsymbol{a}_t|\boldsymbol{s}_t;\boldsymbol{\theta}_t)$.
			\STATE Map the elements $\hat{a}_{i,t}$ to $a_{i,t}$ by Eq. (\ref{Eq:quantify}).
			\STATE Take the after-mapped action $\boldsymbol{a}_t$.
			\STATE Obtain reward $r_t$ by Eq. (\ref{eq:rereward1}).
			\STATE Collect the next state $\boldsymbol{s}_{t+1}$.
			%\STATE Pass $r_t$ and $\boldsymbol{s}[t\text{+1}]$ to the critic.
			\STATE Approximate the value functions $V^{\boldsymbol{\theta}}(\boldsymbol{s}_t;\boldsymbol{\omega}_t)$ and $V^{\boldsymbol{\theta}}(\boldsymbol{s}_{t+1};\boldsymbol{\omega}_t)$ by the DNN of the critic.
			\STATE Calculate TD error $\delta_{_V}(\boldsymbol{\omega}_t)$ by Eq. (\ref{eq:TD_V}).
			\STATE Form and store a new tuple $\{\boldsymbol{s}_t,\boldsymbol{s}_{t+1},r_t,\delta_{_V}(\boldsymbol{\omega}_t)\}$.
			\STATE Obtain $\boldsymbol{\theta}_{t+1}$ and $\boldsymbol{\omega}_{t+1}$ by gradient descent.
			\STATE $\boldsymbol{s}_{t}=\boldsymbol{s}_{t+1}$; $\boldsymbol{\theta}_{t}=\boldsymbol{\theta}_{t+1}$; $\boldsymbol{\omega}_{t}=\boldsymbol{\omega}_{t+1}$.
	\ENDFOR
	\ENDFOR
  \end{algorithmic}
\end{algorithm}

\subsubsection{Action Space Restriction to Improve Solution Quality}
Although AC-DSOS can tackle the issue of discrete action space by the above mapping operation, exploring in a huge space remains difficult.
To improve the exploration efficiency and the quality of the solution, we design a method to restrict the action space in the learning process (line 5 in Alg. \ref{alg:AC}).
At the beginning of each frame, we first observe which users' demands have been satisfied.
Then, we remove the corresponding candidate groups, i.e., the groups containing the successfully served users.
Therefore, the size of the action space is not fixed over $T_{max}$ but gradually decreases.
The action space restriction can help the agent to avoid redundant searches for demand-satisfied users.
Besides, searching in a smaller action space speeds up the algorithm to converge, thereby improving search efficiency and quality.

\subsubsection{Re-designed Reward Function to Deal with Feasibility Issues}
Without a carefully designed mechanism, the actions made in conventional AC-DRL may easily violate constraints, thus fail to guarantee the solution feasibility. 
In $\mathcal{P}_1$, the major difficulty comes from constraints (\ref{eq:demands}), whereas (\ref{eq:order1})-(\ref{eq:varscons4}) can be satisfied by properly defined actions.
Under the commonly-used reward designs, e.g., Eq. (\ref{eq:reward1}) or Eq. (\ref{eq:reward2}), constraint (\ref{eq:demands}) may not be satisfied since the criterion of the decision making is to minimize the objective energy without considering constraints.
To solve the problem, we re-design the reward function by incorporating constraint (\ref{eq:demands}), which is given by:
\begin{align}
\label{eq:rereward1}
r_t = \frac{\sum_{n=1}^{N}{d}^{\pi}_{n,t}}{(e^{\pi}_t)^\epsilon}.
\end{align}
The rationale is that the proposed reward function is the ratio between the delivered data and the consumed energy on frame $t$, where $\epsilon$ is a control parameter.
When $\epsilon$ is small, the reward enforces the UAV to deliver more data to meet users' demands.
However, transmitting more data results in more energy consumption.
To control energy growth, we can increase $\epsilon$ such that the agent will reduce the energy consumption to avoid the reward losses.
Thus, by tuning an appropriate $\epsilon$, the decisions made by AC-DSOS can achieve good energy-saving performance while satisfying users' demands.

\section{Numerical Results}\label{sec:5}
In this section, we present numerical results to evaluate the performance of the proposed AC-DSOS algorithm and compare it with other schemes:
\begin{itemize}
\item
Previous AC-DRL scheme: Deep deterministic policy gradient (DDPG) \cite{davidDDPG};
\item
High-complexity near-optimal scheme: the proposed GSS-HEU in Alg. \ref{alg:Multiple};
\item
Low-complexity sub-optimal scheme: semi-orthogonal user scheduling-based heuristic algorithm (SUS-HEU) \cite{Yoosus};
\item
Optimal scheme: relax-and-approximate approach.
\end{itemize}
DDPG provides performance benchmarks from a typical actor-critic perspective, where a deterministic policy is applied without action space restriction.
The structure of the DNNs, parameter settings, and reward function Eq. (\ref{eq:rereward1}) for AC-DSOS and DDPG are the same in order to enable a feasible solution from DDPG.
The proposed sub-optimal GSS-HEU and optimal algorithms, and sub-optimal SUS-HEU in \cite{Yoosus} benchmark AC-DSOS from an optimization aspect, where SUS-HEU adopts a simple user-grouping strategy with lower complexity than GSS-HEU. 

In the simulation, we first evaluate the performance of energy consumption and computational time.
After that, we justify the developed new reward function in guaranteeing solution feasibility by comparing several well-known reward functions.
Furthermore, we evaluate the convergence performance of AC-DSOS with different learning rates.

\subsection{Parameter Settings}
The UAV is equipped with $L=10$ antennas serving $N=3$ clusters.
The ground users are randomly scattered in the service area.
Each cluster contains up to $K=9$ users.
The users' demands $q_{k,n}$ are randomly selected from $\{1, 2, 3, 4, 5\}$ (Mbit).
%Each cluster with up to 10 users.
We assume the bandwidth $B=10$ MHz, noise power $\sigma^2=0.1$ mW, hovering power $P_{H}=10$ W, and transmit power $p_{k,g,n}=3$ W, referring to \cite{zengenergyuav}.
Based on FSMC, we quantize $\beta^{_{(kk)}}_{^{g,n,t}}$ and $\beta^{_{(kj)}}_{^{g,n,t}}$ into 9 levels, $\{0,0.3,0.6,0.9,1.2,1.5,1.8,2.1,2.4\}$.
The setting of the transfer probability matrix is similar in \cite{Hedrl}. 
Two fully-connected DNNs are employed as the actor and the critic.
The adopted parameters for implementing AC-DSOS are summarized in Table \ref{tab:DNN}.
\begin{table}[h]
\small
\vspace{-0mm}
\centering
\caption{Parameters in AC-DSOS}
\label{tab:DNN}
\begin{tabular}{|c||c|c|}
\hline
Parameters & Actor & Critic \\
\hline
%\multirow{2}*{\tabincell{c}{Number of\\hidden layers}} & \multirow{2}*{3}  \\
Number of hidden layers & 3 & 3\\
%{} & {}\\
\hline
Number of nodes/layer & 300 & 300 \\
\hline
%\multirow{2}*{\tabincell{l}{Activation function for hidden layers}} & \multirow{2}*{ReLU} & \multirow{2}*{ReLU} \\
Activation function (hidden layers) & ReLU & ReLU \\
\hline
Activation function (output layer) &  Sigmoid & None \\
\hline
Learning rate & 0.003 & 0.002\\
\hline
Loss function & Eq. (\ref{eq:loss_actor}) & Eq. (\ref{eq:loss_TD_V})\\
\hline
Optimizer & Adam & Adam\\
\hline
Batch size & 64 & 64\\
\hline
Discount factor $\gamma$ & \multicolumn{2}{c|}{0.9}\\
\hline
Memory size & \multicolumn{2}{c|}{10,000 tuples}\\
\hline
Number of learning episodes & \multicolumn{2}{c|}{400}\\
\hline
Value range $[-\kappa, \kappa]$  of $\hat{a}_{i,t}$ & \multicolumn{2}{c|}{[-2, 2]}\\
\hline
\multirow{2}*{Software platform} &  \multicolumn{2}{c|}{Python 3.6 with}\\
& \multicolumn{2}{c|}{TensorFlow 0.12.1}\\
\hline
\end{tabular}
\end{table}

\subsection{Results and Analysis}
Firstly, by comparing with four benchmarking algorithms in Fig. \ref{fig:energy_alg}, the proposed AC-DSOS achieves a good trade-off between energy minimization and computational time.
Note that for $K > 7$, the optimal energy results are absent due to the high complexity and the corresponding long computational time.
From Fig. \ref{fig:energy_alg}, AC-DSOS saves around 29.94\% energy compared to DDPG in average.
Overall, AC-DSOS provides a sub-optimal solution, with 19.17\% gap to the optimum.
GSS-HEU achieves near optimality, and consumes less 9.8\% energy than AC-DSOS in average but with paying much higher complexity and time, e.g., see Fig. \ref{fig:comp_time}.
SUS-HEU consumes the highest energy since it schedules users based on channel conditions without considering energy consumption.
It is also shown that the total objective energy follows a roughly linear increase in all the algorithms.
The gaps between the optimal algorithm and other algorithms become larger as $K$ increases.
When $K$ grows from 5 to 7, the gap to the optimum increases from 47.7\% to 65.1\% for SUS-HEU, and from 31\% to 44.5\% for DDPG.
In AC-DSOS, since the delivery-completed users are deleted during the learning process, the size of the action space will continuously decrease.
This improves the searching efficiency and quality, and reduces the growth rate of the gap as $K$ increases, from 11.1\% ($K=5$) to 16.7\% ($K=7$).

\begin{figure}[h]
\begin{center}
\centering
\vspace{-0.2cm}  %调整图片与上文的垂直距离
\includegraphics[scale=0.53]{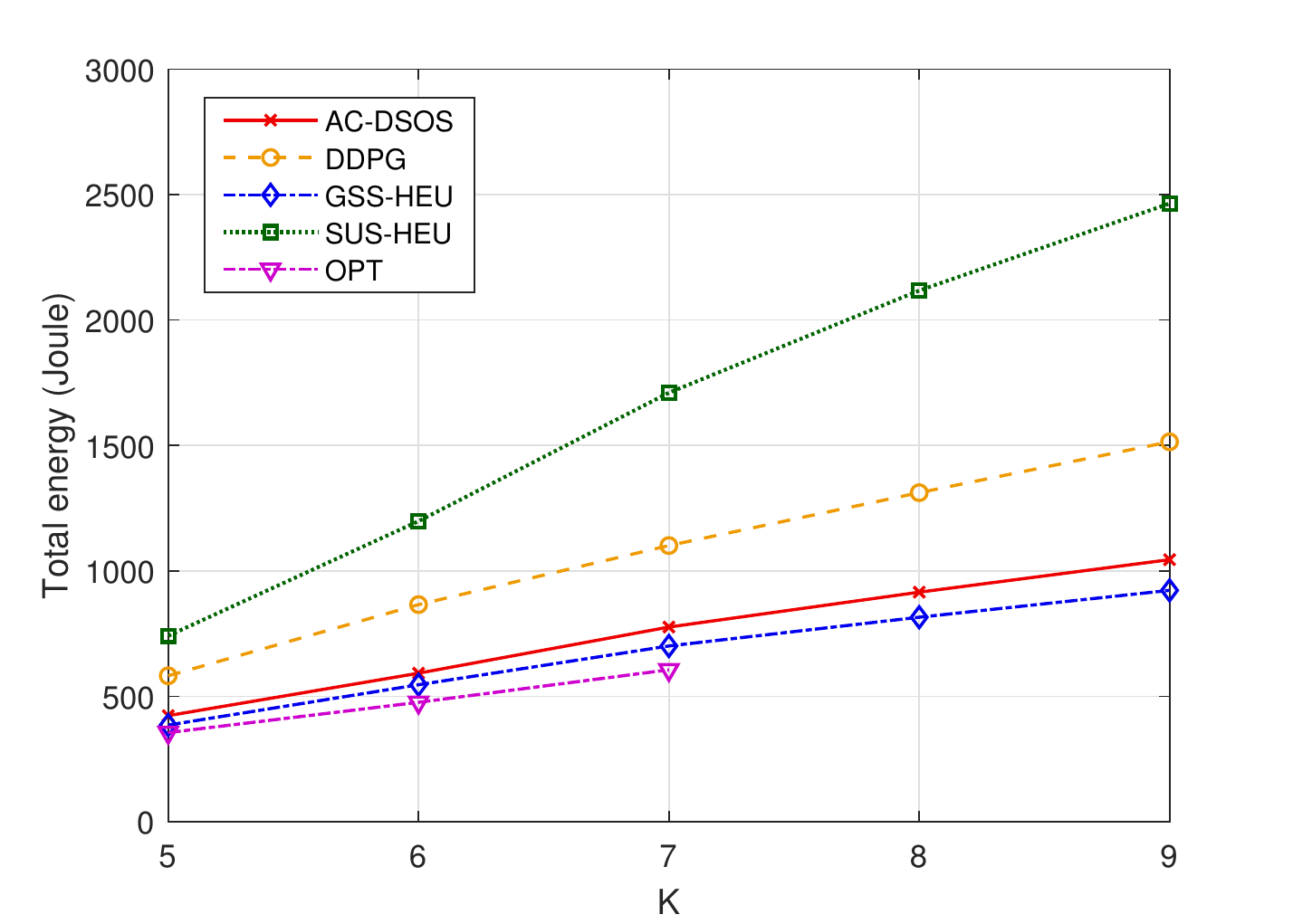}
\captionsetup{font={small}}
\caption{Total energy vs. $K$ ($T_{max}=160$).}
\setlength{\belowcaptionskip}{-0.3cm}   %调整图片标题与下文距离
\label{fig:energy_alg}
\end{center}
\end{figure}

\begin{figure}[h]
\begin{center}
\centering
\vspace{-0.5cm}  %调整图片与上文的垂直距离
\includegraphics[scale=0.53]{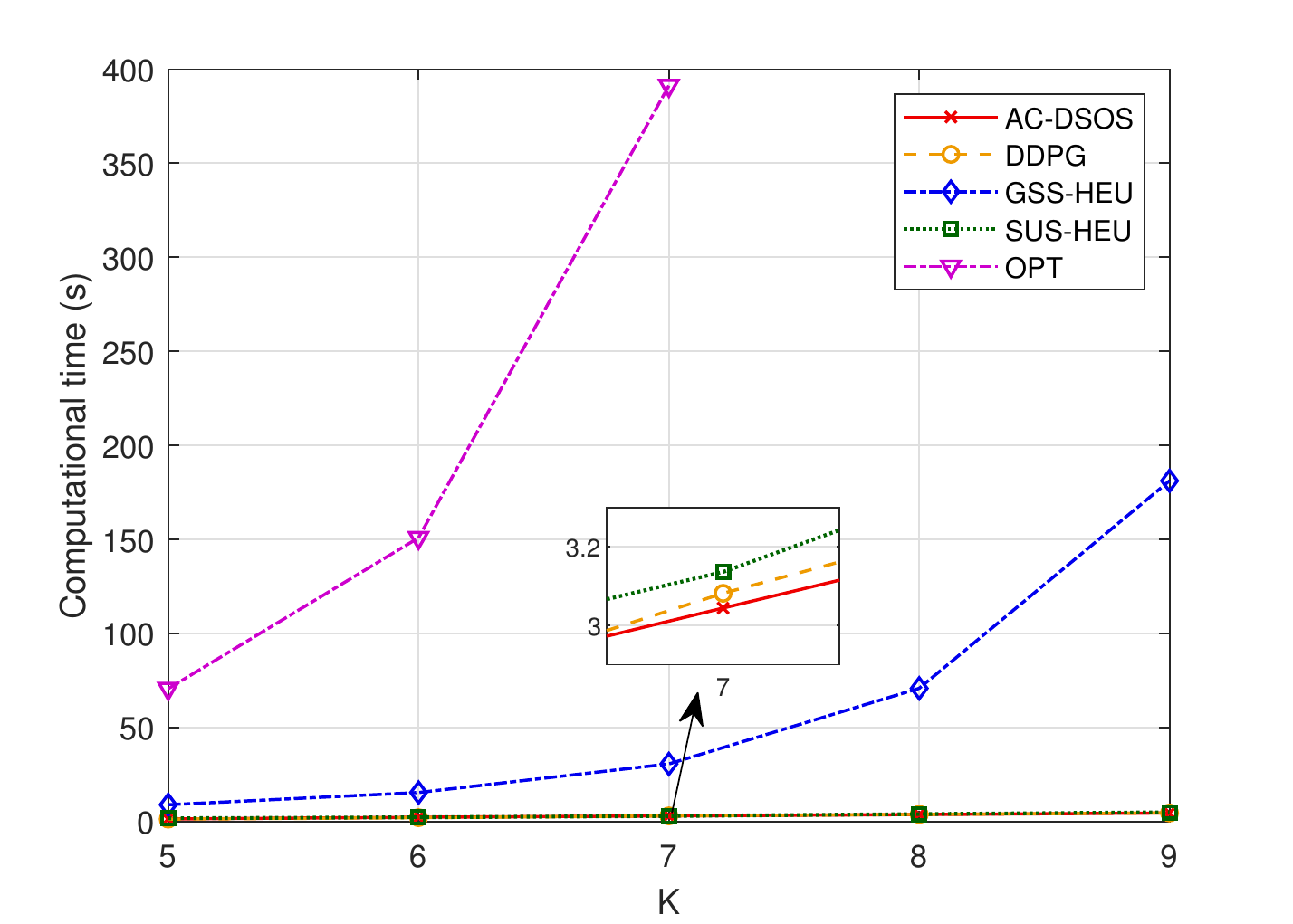}
\setlength{\belowcaptionskip}{-0.0cm}   %调整图片标题与下文距离
\captionsetup{font={small}}
\caption{Average computational time vs. $K$ ($T_{max}=160$).}
\label{fig:comp_time}
\end{center}
\end{figure}

Fig. \ref{fig:comp_time} compares the computational time with respect to $K$.
The computational time records from giving inputs to algorithms until returning the optimized results.
In GSS-HEU and the optimal approach, the computational time grows exponentially with $K$, whereas the proposed AC-DSOS along with DDPG and SUS-HEU maintain at a low magnitude and insensitive to the increase of $K$.
AC-DSOS saves 99.23\% and 92.86\% computational time compared to the optimal algorithm and the GSS-HEU when $K=7$.
This is due to the fact that DRL can provide online decisions based on the current environment state instead of solving the optimization problem directly.
The computational time of AC-DSOS slightly lower than DDPG and SUS-HEU.
However, by recalling Fig. \ref{fig:energy_alg}, AC-DSOS saves 29.94\% and 52.51\% energy compared with DDPG and SUS-HEU, respectively.

\begin{figure}[h]
\begin{center}
\centering
\vspace{-0.2cm}  %调整图片与上文的垂直距离
\includegraphics[scale=0.53]{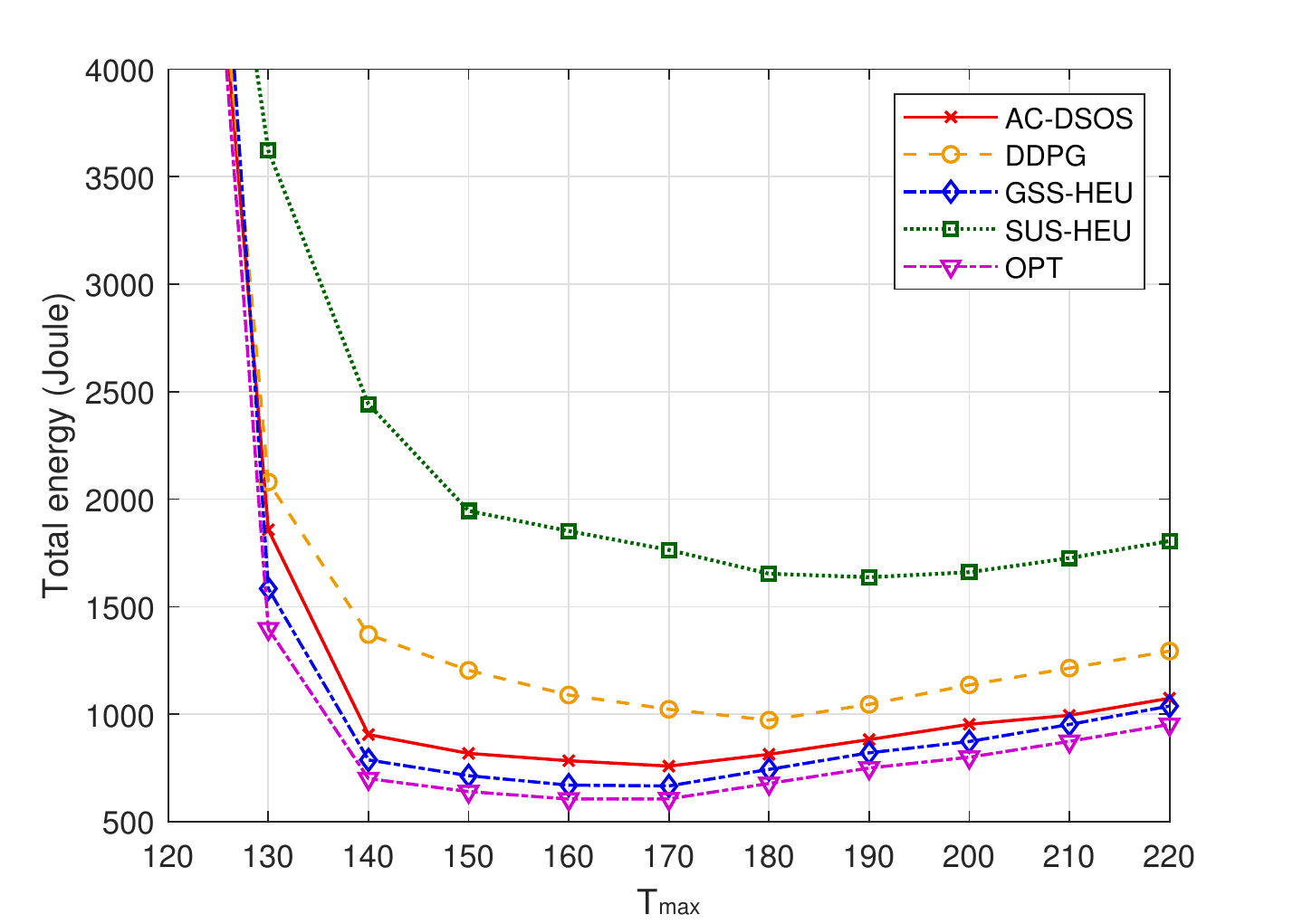}
\captionsetup{font={small}}
\caption{Total energy vs. $T_{max}$ $(K=7)$.}
\setlength{\belowcaptionskip}{-0.3cm}   %调整图片标题与下文距离
\label{fig:totalenergy}
\end{center}
\end{figure}

\begin{figure}[h]
\begin{center}
\centering
\vspace{-0.5cm}  %调整图片与上文的垂直距离
\includegraphics[scale=0.51]{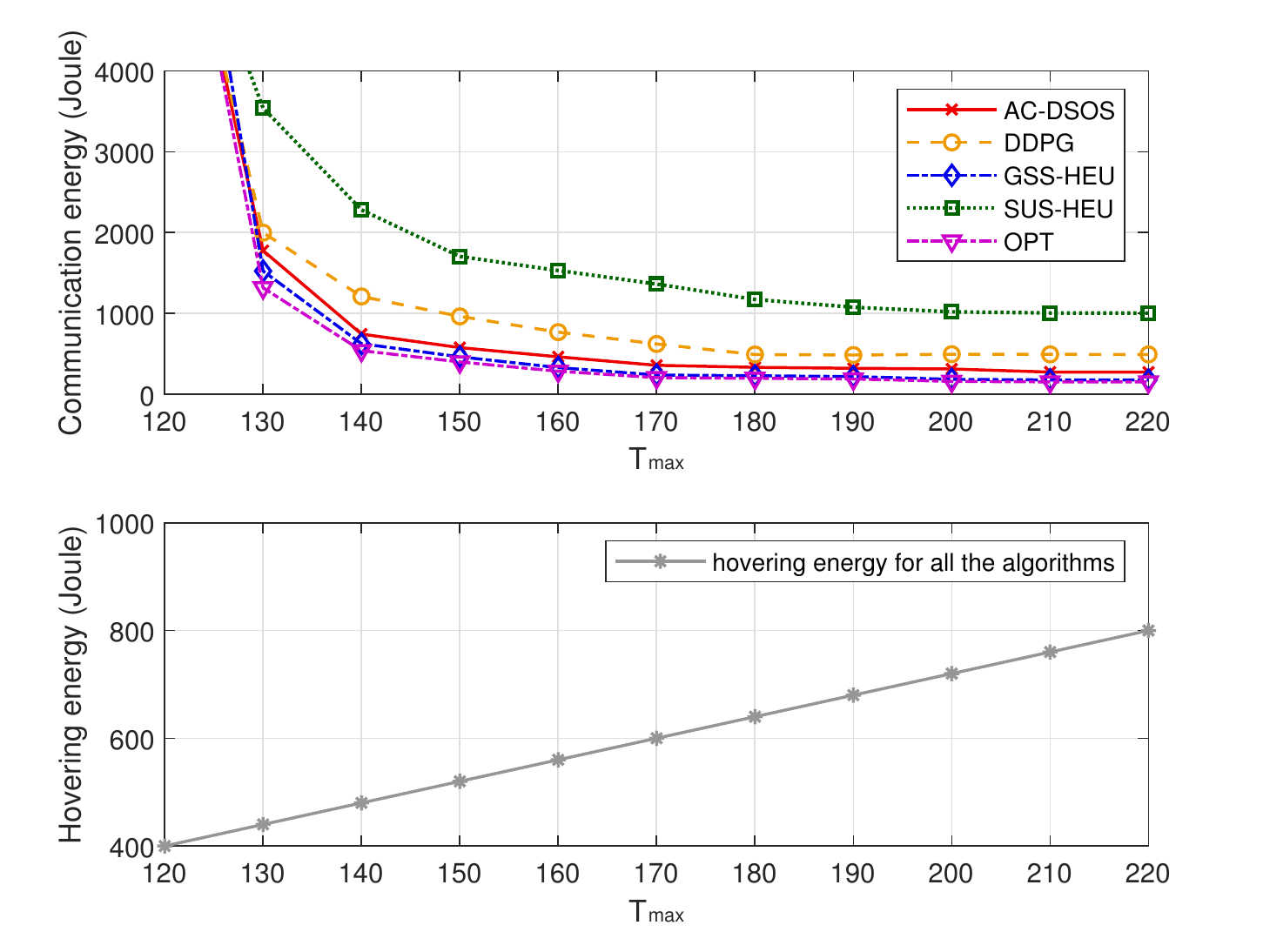}
\setlength{\belowcaptionskip}{-0.0cm}   %调整图片标题与下文距离
\captionsetup{font={small}}
\caption{Communication and hovering energy vs. $T_{{max}}$ $(K=7)$.}
\label{fig:comhovenergy}
\end{center}
\end{figure}

Fig. \ref{fig:totalenergy} demonstrates the total energy consumption with respect to $T_{max}$, and Fig. \ref{fig:comhovenergy} illustrates the communication energy and hovering energy separately.
From Fig. \ref{fig:totalenergy}, AC-DSOS outperforms DDPG by saving 21.37\% total energy in average.
The average gap between GSS-HEU and the optimal solution is 8.91\% smaller than that of AC-DSOS, but, from Fig. \ref{fig:comp_time}, GSS-HEU consumes nearly 126 times higher calculation time than AC-DSOS at $T_{max}=160$.
The energy-saving performance of SUS-HEU is worse than other algorithms and its gap to the optimum reaches  59.44\%.
Fig. \ref{fig:totalenergy} also shows that, as $T_{max}$ increases, the objective energy rapidly decreases first then grows steadily.
This can be explained via Fig. \ref{fig:comhovenergy}.
The objective energy consists of the communication energy and hovering energy.
From Fig. \ref{fig:comhovenergy}, the communication energy drops rapidly when $T_{max}<140$, and becomes stable after $T_{max}>180$.
Whereas, the hovering energy increases linearly with $T_{max}$ for all the algorithms.

Fig. \ref{fig:outagerewards} verifies the capability of the proposed reward function in dealing with feasibility issues, where a feasible solution is obtained only if the ratio of delivered demand over total demand in \textit{y}-axis achieves 100\%.
From Fig. \ref{fig:outagerewards}, the reward functions used in Eq. (\ref{eq:reward1}) and Eq. (\ref{eq:reward2}) fail to guarantee the feasibility of the solution.
For the re-designed reward, we evaluate the performance by setting $\epsilon$ to 1, 1.2, and 1.5.
A small $\epsilon$ means that transmitting more data can bring more rewards gain than saving energy.
When $\epsilon$ drops below 1.2, the feasibility issue can be solved.
Fig. \ref{fig:energyrewards} shows the objective energy with different $\epsilon$. 
It can be found that a smaller $\epsilon$ leads to more energy consumption.
Thus, an appropriate parameter $\epsilon$ lies at 1.2, enabling the after-learned solution to guarantee the demands while consuming less energy.

\begin{figure}[h]
\begin{center}
\centering
\vspace{-0.2cm}  %调整图片与上文的垂直距离
\includegraphics[scale=0.53]{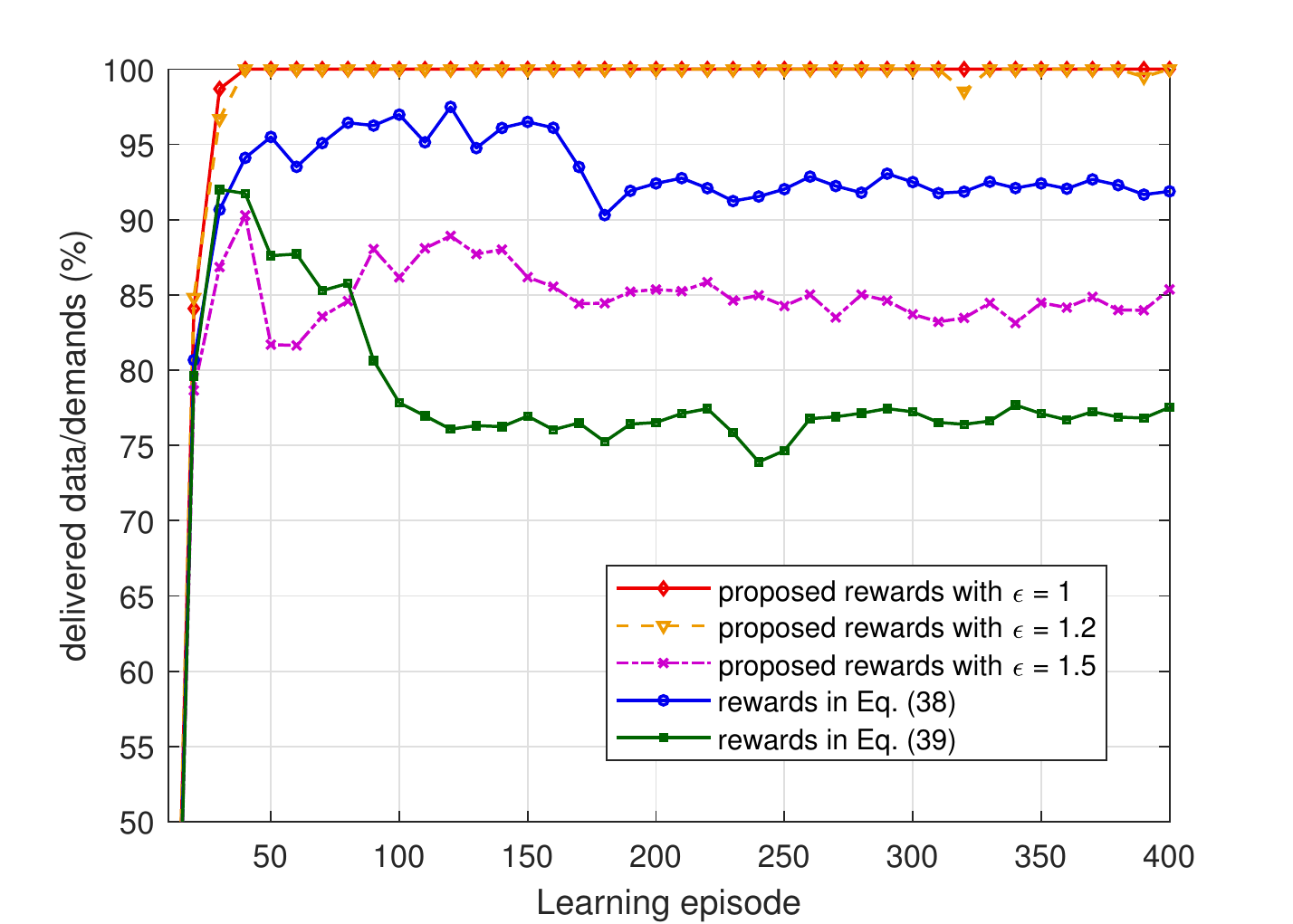}
\setlength{\belowcaptionskip}{-0.3cm}   %调整图片标题与下文距离
\captionsetup{font={small}}
\caption{Feasibility vs. learning episode.}
\label{fig:outagerewards}
\end{center}
\end{figure}

\begin{figure}[h]
\begin{center}
\centering
\vspace{-0.1cm}  %调整图片与上文的垂直距离
\includegraphics[scale=0.53]{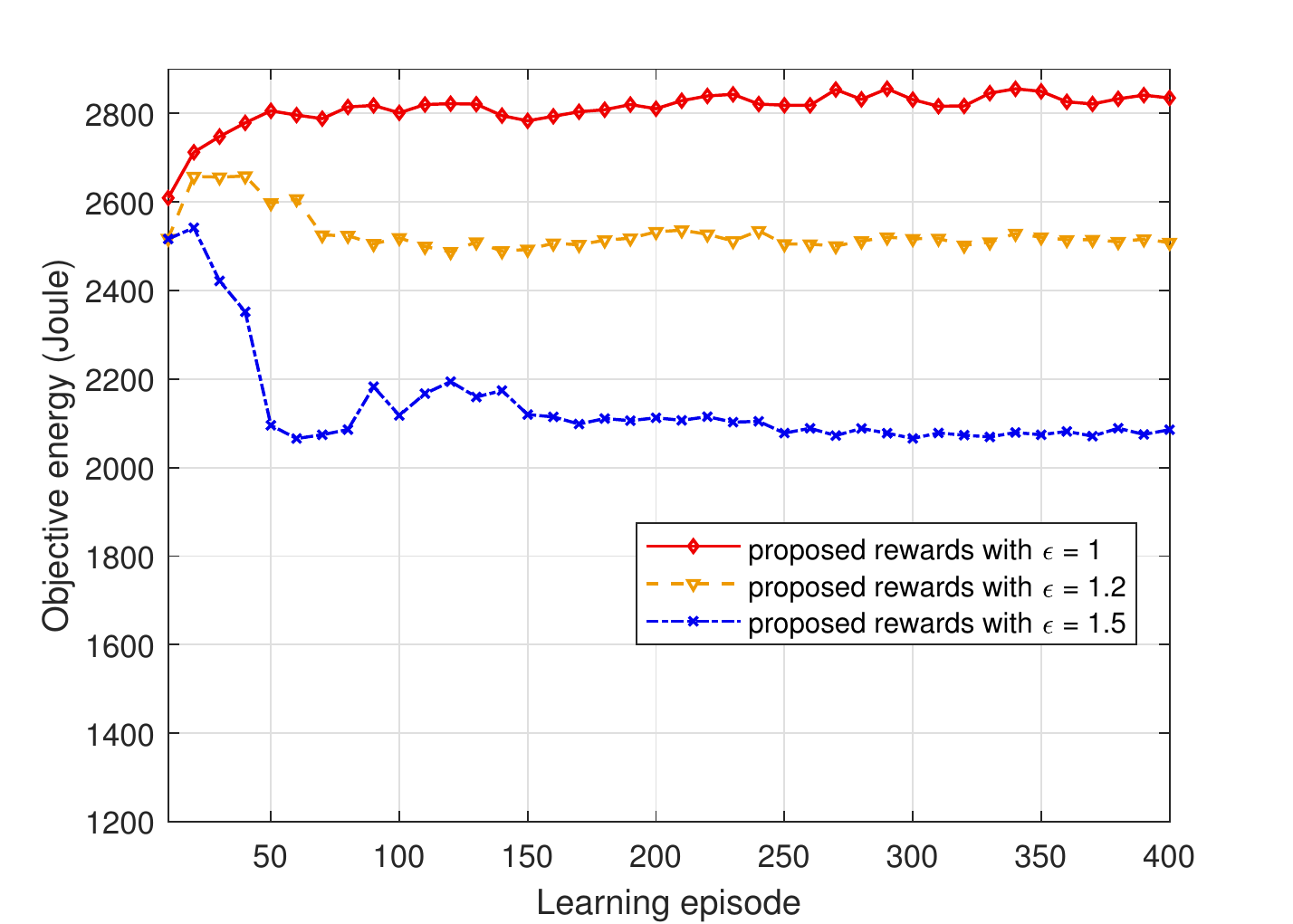}
\setlength{\belowcaptionskip}{-0.3cm}   %调整图片标题与下文距离
\captionsetup{font={small}}
\caption{Energy vs. learning episode.}
\label{fig:energyrewards}
\end{center}
\end{figure}

\begin{figure}[h]
\begin{center}
\centering
\vspace{-0.1cm}  %调整图片与上文的垂直距离
\includegraphics[scale=0.53]{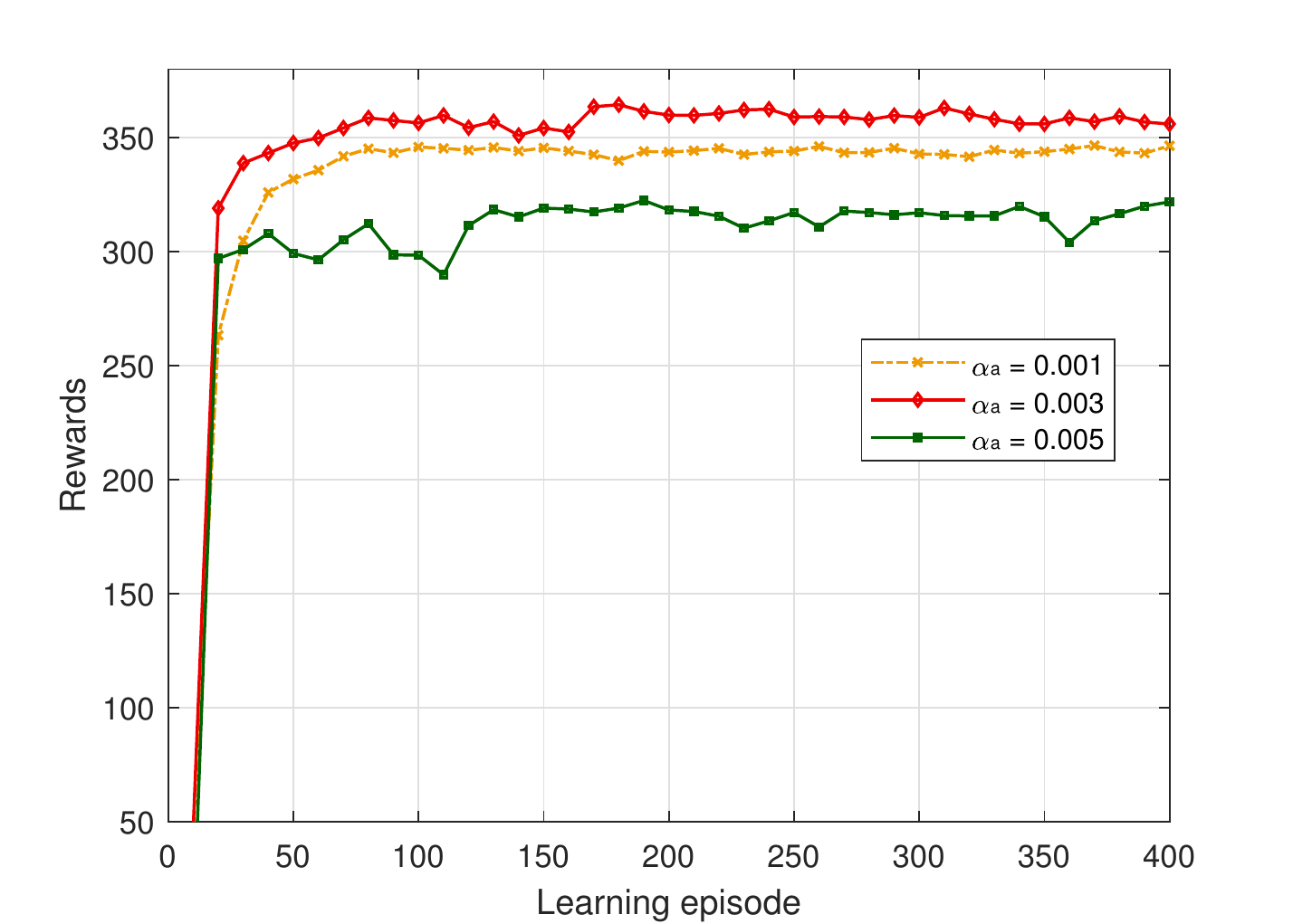}
\setlength{\belowcaptionskip}{-0.1cm}   %调整图片标题与下文距离
\captionsetup{font={small}}
\caption{Rewards vs. learning episode.}
\label{fig:convergence}
\end{center}
\end{figure}

Fig. \ref{fig:convergence} demonstrates the convergence of AC-DSOS with different actor's learning rate $\alpha_a$.
The \textit{x}-axis is the learning episode and the \textit{y}-axis is the reward value.
When $\alpha_a$ increases from 0.001 to 0.003, the reward value grows by 2.8\% at the convergence.
If $\alpha_a$ increases to 0.005, the curve fluctuates and the converged value is 7.2\% lower than the case of $\alpha_a = 0.001$.
Taking the actor as an example, the learning rate for the critic $\alpha_c$ has the same tendency.
In conclusion, the learning rates of the actor and critic are sensitive to the convergence, and need to be properly selected, e.g., 0.003 for the actor.

\section{Conclusion}\label{sec:6}
In this paper, we have investigated an energy minimization problem for UAV-aided communication systems from the perspective of AC-DRL.
The formulated problem is combinatorial and non-convex.
We provided an optimal relax-and-approximate method and proposed a GSS-based heuristic algorithm to solve the problem and serve as benchmarks.
To make the solutions adaptive to online operation, we propose an AC-DSOS algorithm.
Different from previous AC-DRL methods, the proposed AC-DSOS is able to deal with the huge discrete action space and guarantee the feasibility.
Numerical results have shown that AC-DSOS provides a good trade-off between energy efficiency and computational efficiency.
Furthermore, the re-designed reward function is effective to deal with the feasibility issue. 
An extension of the current work is to jointly optimize energy consumption in uplink, downlink communications, and UAV propulsion, e.g., UAV downlink data-delivery and uplink data-collection tasks co-exist among clusters. 
 %\section*{Acknowledgment}


\ifCLASSOPTIONcaptionsoff
  \newpage
\fi

%\newpage

\end{document}